\documentclass[runningheads]{llncs}

\usepackage{booktabs} 
\usepackage{subcaption} 
\usepackage{bussproofs}
\usepackage[cal=boondoxo]{mathalfa}
\DeclareMathAlphabet{\mathpzc}{OT1}{pzc}{m}{it}
\usepackage{amsmath}
\usepackage{color}
\usepackage{xspace}
\usepackage{orcidlink}

\newcommand{\ignore}[1]{{}}

\newcommand{\syntaxDef}[3]{\rulebox{%
\syntaxKeyword$#1\mathrel{::=}{#2}$ \ifthenelse{\equal{#3}{}}{}{[#3]}%
}%
}

\makeatletter
\newcommand{\shorteq}{%
  \settowidth{\@tempdima}{-}% Width of hyphen
  \resizebox{\@tempdima}{\height}{=}%
}
\makeatother

\usepackage{amssymb}
\usepackage{xcolor}
\usepackage{colortbl}

\usepackage[utf8]{inputenc}

\usepackage{amsmath}
\usepackage{stmaryrd}
\usepackage{hyperref}

\usepackage{subcaption}
\usepackage{wrapfig}

\usepackage{xspace}
\usepackage{float}
\usepackage{balance}

\usepackage{textcomp} %just for a vertical quote

\usepackage{changepage}
\usepackage{array,etoolbox}

\preto\tabular{\setcounter{magicrownumbers}{0}}
\newcounter{magicrownumbers}

\usepackage{cleveref} % must be loaded after hyperref and amsmath

\usepackage{wrapfig}
\usepackage{caption}

\usepackage{longtable}
\usepackage{tabu}

\usepackage{wasysym}
\usepackage{mathpartir}
\usepackage{mathtools} %psmallmatrix
\usepackage{xfrac} % fractions

\usepackage[qm,braket]{qcircuit}

\newcommand{\modmult}[3]{(#1 * #2)\cn{\,\%} \, #3}
\newcommand{\modexp}[3]{{#1}^{#2}\,\cn{\%} \, #3}
\newcommand{\qbool}[4]{\cn{(}#1\,\cn{#2}\,#3\cn{)\,\cn{@}\,}#4}

\newcommand{\mmod}{\cn{\%}}

\colorlet{kwd}{black!80!green}
\definecolor{spec1}{RGB}{78, 131, 162}
\definecolor{spec0}{RGB}{66, 102, 136}
\definecolor{lespec}{RGB}{30, 80, 180}
\colorlet{spec}{lespec}
\colorlet{auto}{lespec!35!lightgray}
\colorlet{stack}{magenta}

\newcommand{\qafny}{\rulelab{Qafny}\xspace}

\newcommand{\sourcelang}{\ensuremath{\mathcal{O}\textsc{qimp}}\xspace}
\newcommand{\oqasm}{\textsc{Oqasm}\xspace}

\newcommand{\pqasm}{\ensuremath{\textsc{Pqasm}}\xspace}

\newcommand{\sqir}{SQIR\xspace}

\newcommand{\voqc}{\textsc{VOQC}\xspace}
\newcommand{\tket}{t$\vert$ket$\rangle$\xspace}
\newcommand{\myparagraph}[1]{\noindent\paragraph{\textbf{#1}}}

\usepackage{tikz}

\newcommand{\aket}[2]{\ket{#1}_{\textcolor{spec}{#2}}}
\newcommand{\qket}[1]{\ket{\Delta(#1)}}

\newcommand{\slen}[1]{|#1|}
\newcommand{\CC}{\ensuremath{\mathbb{C}}\xspace}
\usepackage{pgf}

\usepackage{tikz} % circuit diagrams
\usetikzlibrary{%
  arrows,%
  shapes.misc,% wg. rounded rectangle
  shapes.geometric, % diamonds
  shapes.arrows,%
  shapes.callouts,
  shapes.gates.logic.US,
  chains,%
  matrix,%
  positioning,% wg. " of "
  scopes,%
  decorations.pathmorphing,% /pgf/decoration/random steps | erste Graphik
  decorations.text,
  decorations.pathreplacing, % braces
  shadows,%
  automata,
  fit, calc, arrows.meta
}

\tikzset{ machine/.style={
    rectangle,
    minimum width=25mm,
    minimum height=18mm,
    text width=24mm,
    align=center,
    very thick,
    draw=black,
    color=black,
    fill=white,
  }
}

\DeclarePairedDelimiter\abs{\lvert}{\rvert}
\DeclarePairedDelimiter\norm{\lVert}{\rVert}

\makeatletter
\let\oldabs\abs
\def\abs{\@ifstar{\oldabs}{\oldabs*}}
\let\oldnorm\norm
\def\norm{\@ifstar{\oldnorm}{\oldnorm*}}
\makeatother

\makeatletter
\DeclareRobustCommand{\vardivision}{%
  \mathbin{\mathpalette\@vardivision\relax}% 
}
\newcommand{\@vardivision}[2]{%
  \reflectbox{$\m@th\smallsetminus$}%
}
\makeatother

\DeclareFontFamily{U} {MnSymbolC}{}
\DeclareFontShape{U}{MnSymbolC}{m}{n}{
  <-6> MnSymbolC5
  <6-7> MnSymbolC6
  <7-8> MnSymbolC7
  <8-9> MnSymbolC8
  <9-10> MnSymbolC9
  <10-12> MnSymbolC10
  <12-> MnSymbolC12}{}
\DeclareFontShape{U}{MnSymbolC}{b}{n}{
  <-6> MnSymbolC-Bold5
  <6-7> MnSymbolC-Bold6
  <7-8> MnSymbolC-Bold7
  <8-9> MnSymbolC-Bold8
  <9-10> MnSymbolC-Bold9
  <10-12> MnSymbolC-Bold10
  <12-> MnSymbolC-Bold12}{}

\DeclareSymbolFont{MnSyC} {U} {MnSymbolC}{m}{n}

\DeclareMathSymbol{\sqcupplus}{\mathbin}{MnSyC}{70}

\usepackage{booktabs}

\usepackage{alltt}
\usepackage{listings,lstcoq}
\definecolor{ltblue}{rgb}{0,0.4,0.4}
\definecolor{dkblue}{rgb}{0,0.1,0.6}
\definecolor{dkgreen}{rgb}{0,0.35,0}
\definecolor{dkviolet}{rgb}{0.3,0,0.5}
\definecolor{dkred}{rgb}{0.5,0,0}
\usepackage[export]{adjustbox}

\newcommand{\code}[1]{{\small\texttt{#1}}}
 \newcommand{\rulelab}[1]{{\small \textsc{#1}}}

\newcommand{\sskip}{\texttt{\{\}}}

\newcommand{\sifb}[3]{\texttt{if}~{(#1)}~{#2}~\texttt{else}~{#3}}

\newcommand{\tnor}[1]{\texttt{Nor}({#1})}
\newcommand{\tnort}{\texttt{Nor}}
\newcommand{\trot}[1]{\texttt{Rot}({#1})}
\newcommand{\trott}{\texttt{Rot}}
\newcommand{\thad}[1]{\texttt{Had}(#1)}
\newcommand{\thadt}{\texttt{Had}}

\DeclareMathOperator*{\Motimes}{\text{\raisebox{0.25ex}{\scalebox{0.8}{$\bigotimes$}}}}
\DeclareMathOperator*{\Msum}{\text{\raisebox{0.25ex}{\scalebox{0.8}{$\sum$}}}}

\newcommand{\sminus}{\texttt{-}}
\newcommand{\splus}{\texttt{+}}

\newcommand{\sseq}[2]{{#1}\,\texttt{;}\,{#2}}

\newcommand{\insta}[3][ ]{\texttt{#2}^{#1}~{#3}}
\newcommand{\insttwo}[4][ ]{\texttt{#2}^{#1}~{#3}~{#4}}

\newcommand{\inot}[1]{\insta{X}{#1}}
\newcommand{\ictrl}[2]{\insttwo{CU}{#1}{#2}}
\newcommand{\iadd}[2]{\cn{add}(#1,#2)}
\newcommand{\irz}[3][ ]{\insttwo[#1]{RZ}{#2}{#3}}
\newcommand{\isr}[3][ ]{\insttwo[#1]{SR}{#2}{#3}}

\newcommand{\ilshift}[1]{\insta{Lshift}{#1}}
\newcommand{\irshift}[1]{\insta{Rshift}{#1}}
\newcommand{\irev}[1]{\insta{Rev}{#1}}
\newcommand{\iqft}[3][ ]{\insttwo[#1]{QFT}{#2}{#3}}

\newcommand{\tphi}[1]{\texttt{Phi}~{#1}}
\newcommand{\ihad}[1]{\texttt{H}(#1)}
\newcommand{\inew}[1]{\texttt{new}(#1)}
\newcommand{\iry}[2]{\texttt{Ry}^{#1}{#2}}

\newcommand{\hsp}[1]{\mathcal{#1}}
\newcommand{\iseq}[2]{{#1}\,\texttt{;}\,{#2}}

\newcommand{\itext}[1]{\texttt{#1}}

\newcommand{\instr}{\iota}
\newcommand{\iskip}[1]{\texttt{SKIP}\,{#1}}

\newcommand{\app}[3]{#2\texttt{[}{#3}\mapsto{#1}\texttt{]}}
\newcommand{\xsem}{\texttt{xg}}

\newcommand{\qsem}{\texttt{qt}}
\newcommand{\psem}{\texttt{pm}}
\newcommand{\rsem}{\texttt{rz}}
\newcommand{\rrsem}{\texttt{rrz}}

\newcommand{\csem}{\texttt{cu}}

\newcommand{\cn}[1]{\texttt{#1}}
\newcommand{\Omegasz}{\Sigma}
\newcommand{\Omegaty}{\Omega}
\newcommand{\denote}[1]{\llbracket #1 \rrbracket\xspace}

\newcommand{\qfun}[2]{#1\langle #2 \rangle}

\newcommand{\smea}[3]{\texttt{let}\;{#1}\;\texttt{=}\;{\mathpzc{M}\cn{(}#2\cn{)}}\;\texttt{in}\;#3}

\usepackage{newunicodechar}
\let\Alpha=A
\let\Beta=B
\let\Epsilon=E
\let\Zeta=Z
\let\Eta=H
\let\Iota=I
\let\Kappa=K
\let\Mu=M
\let\Nu=N
\let\Omicron=O
\let\omicron=o
\let\Rho=P
\let\Tau=T
\let\Chi=X

\newunicodechar{Α}{\ensuremath{\Alpha}}
\newunicodechar{α}{\ensuremath{\alpha}}
\newunicodechar{Β}{\ensuremath{\Beta}}
\newunicodechar{β}{\ensuremath{\beta}}
\newunicodechar{Γ}{\ensuremath{\Gamma}}
\newunicodechar{γ}{\ensuremath{\gamma}}
\newunicodechar{Δ}{\ensuremath{\Delta}}
\newunicodechar{δ}{\ensuremath{\delta}}
\newunicodechar{Ε}{\ensuremath{\Epsilon}}
\newunicodechar{ε}{\ensuremath{\epsilon}}
\newunicodechar{ϵ}{\ensuremath{\varepsilon}}
\newunicodechar{Ζ}{\ensuremath{\Zeta}}
\newunicodechar{ζ}{\ensuremath{\zeta}}
\newunicodechar{Η}{\ensuremath{\Eta}}
\newunicodechar{η}{\ensuremath{\eta}}
\newunicodechar{Θ}{\ensuremath{\Theta}}
\newunicodechar{θ}{\ensuremath{\theta}}
\newunicodechar{ϑ}{\ensuremath{\vartheta}}
\newunicodechar{Ι}{\ensuremath{\Iota}}
\newunicodechar{ι}{\ensuremath{\iota}}
\newunicodechar{Κ}{\ensuremath{\Kappa}}
\newunicodechar{κ}{\ensuremath{\kappa}}
\newunicodechar{Λ}{\ensuremath{\Lambda}}
\newunicodechar{λ}{\ensuremath{\lambda}}
\newunicodechar{Μ}{\ensuremath{\Mu}}
\newunicodechar{μ}{\ensuremath{\mu}}
\newunicodechar{Ν}{\ensuremath{\Nu}}
\newunicodechar{ν}{\ensuremath{\nu}}
\newunicodechar{Ξ}{\ensuremath{\Xi}}
\newunicodechar{ξ}{\ensuremath{\xi}}
\newunicodechar{Ο}{\ensuremath{\Omicron}}
\newunicodechar{ο}{\ensuremath{\omicron}}
\newunicodechar{Π}{\ensuremath{\Pi}}
\newunicodechar{π}{\ensuremath{\pi}}
\newunicodechar{ϖ}{\ensuremath{\varpi}}
\newunicodechar{Ρ}{\ensuremath{\Rho}}
\newunicodechar{ρ}{\ensuremath{\rho}}
\newunicodechar{ϱ}{\ensuremath{\varrho}}
\newunicodechar{Σ}{\ensuremath{\Sigma}}
\newunicodechar{σ}{\ensuremath{\sigma}}
\newunicodechar{ς}{\ensuremath{\varsigma}}
\newunicodechar{Τ}{\ensuremath{\Tau}}
\newunicodechar{τ}{\ensuremath{\tau}}
\newunicodechar{Υ}{\ensuremath{\Upsilon}}
\newunicodechar{υ}{\ensuremath{\upsilon}}
\newunicodechar{Φ}{\ensuremath{\Phi}}
\newunicodechar{φ}{\ensuremath{\phi}}
\newunicodechar{ϕ}{\ensuremath{\varphi}}
\newunicodechar{Χ}{\ensuremath{\Chi}}
\newunicodechar{χ}{\ensuremath{\chi}}
\newunicodechar{Ψ}{\ensuremath{\Psi}}
\newunicodechar{ψ}{\ensuremath{\psi}}
\newunicodechar{Ω}{\ensuremath{\Omega}}
\newunicodechar{ω}{\ensuremath{\omega}}

\newunicodechar{ℕ}{\ensuremath{\mathbb{N}}}
\newunicodechar{∅}{\ensuremath{\emptyset}}

\newunicodechar{∙}{\ensuremath{\bullet}}
\newunicodechar{≈}{\ensuremath{\approx}}
\newunicodechar{≅}{\ensuremath{\cong}}
\newunicodechar{≡}{\ensuremath{\equiv}}
\newunicodechar{≤}{\ensuremath{\le}}
\newunicodechar{≥}{\ensuremath{\ge}}
\newunicodechar{≠}{\ensuremath{\neq}}
\newunicodechar{∀}{\ensuremath{\forall}}
\newunicodechar{∃}{\ensuremath{\exists}}
\newunicodechar{±}{\ensuremath{\pm}}
\newunicodechar{∓}{\ensuremath{\pm}}
\newunicodechar{·}{\ensuremath{\cdot}}
\newunicodechar{⋯}{\ensuremath{\cdots}}
\newunicodechar{…}{\ensuremath{\ldots}}
\newunicodechar{∷}{~\mathrel{:\!\!\!:}~}
\newunicodechar{×}{\ensuremath{\times}}
\newunicodechar{∞}{\ensuremath{\infty}}
\newunicodechar{→}{\ensuremath{\to}}
\newunicodechar{←}{\ensuremath{\leftarrow}}
\newunicodechar{⇒}{\ensuremath{\Rightarrow}}
\newunicodechar{↦}{\ensuremath{\mapsto}}
\newunicodechar{↝}{\ensuremath{\leadsto}}
\newunicodechar{∨}{\ensuremath{\vee}}
\newunicodechar{∧}{\ensuremath{\wedge}}
\newunicodechar{⊢}{\ensuremath{\vdash}}
\newunicodechar{⊣}{\ensuremath{\dashv}}
\newunicodechar{∣}{\ensuremath{\mid}}
\newunicodechar{∈}{\ensuremath{\in}}
\newunicodechar{⊆}{\ensuremath{\subseteq}}
\newunicodechar{⊂}{\ensuremath{\subset}}
\newunicodechar{∪}{\ensuremath{\cup}}
\newunicodechar{⋓}{\ensuremath{\Cup}}
\newunicodechar{∉}{\ensuremath{\not\in}}
\newunicodechar{√}{\ensuremath{\sqrt}}

\newunicodechar{⊸}{\ensuremath{\multimap}}
\newunicodechar{⊗}{\ensuremath{\otimes}}
\newunicodechar{⨂}{\ensuremath{\bigotimes}}
\newunicodechar{⊕}{\ensuremath{\oplus}}
\newunicodechar{〈}{\ensuremath{\langle}}
\newunicodechar{⟨}{\ensuremath{\langle}}
\newunicodechar{⟩}{\ensuremath{\rangle}}
\newunicodechar{〉}{\ensuremath{\rangle}}
\newunicodechar{¡}{\ensuremath{\upsidedownbang}}
\newunicodechar{∘}{\ensuremath{\circ}}
\newunicodechar{†}{\ensuremath{\dagger}}
\newunicodechar{⊤}{\ensuremath{\top}}
\newunicodechar{⊥}{\ensuremath{\bot}}

\newunicodechar{〚}{\ensuremath{\llbracket}}
\newunicodechar{〛}{\ensuremath{\rrbracket}}

\usepackage{etoolbox} % replaces ifthen package
\newtoggle{comments}
\toggletrue{comments}
  \usepackage[normalem]{ulem}

\iftoggle{comments}{
  \newcommand{\fixme}[1]{\textbf{\textcolor{red}{[ Fixme: #1]}}}
  \newcommand{\todo}[1]{\textbf{\textcolor{green}{[ TODO: #1 ]}}}
  \newcommand{\mwh}[1]{\textbf{\textcolor{red}{[ Mike: #1 ]}}}
  
  \newcommand{\khh}[1]{\textbf{\textcolor{orange}{[ Kesha: #1 ]}}}
  \newcommand{\shh}[1]{\textbf{\textcolor{purple}{[ Shih-Han: #1 ]}}}
  \newcommand{\liyi}[1]{\textbf{\textcolor{blue}{[ Liyi: #1 ]}}}

  \newcommand{\oth}[2]{\textbf{\textcolor{red}{[ #1: #2 ]}}}
  \newcommand{\xwu}[1]{\textbf{\textcolor{purple}{[ Xiaodi: #1 ]}}}
  
  \newcommand{\ynote}[1]{\textbf{\textcolor{magenta}{[ Yi: #1 ]}}}

  \colorlet{MZ}{violet!80!pink}

  \newcommand{\mzr}[1]{{\color{MZ}{#1}}}
  \newcommand{\was}[1]{}
  
  \NewCommandCopy{\Creff}{\Cref}
  \renewcommand{\Cref}[1]{\mbox{\Creff{#1}}}

  \usepackage[inline]{enumitem}
  \colorlet{LC}{cyan!31!teal}

  \usepackage[inline]{enumitem}
}{
  \newcommand{\fixme}[1]{}
  \newcommand{\todo}[1]{}
  \newcommand{\rnr}[1]{}
  \newcommand{\mwh}[1]{}  
  \newcommand{\khh}[1]{}
  \newcommand{\liyi}[1]{}
  \newcommand{\shh}[1]{}
  \newcommand{\xwu}[1]{}
  \newcommand{\oth}[2]{}
  \newcommand{\mzr}[1]{}

  \newcommand{\ynote}[1]{}
}

\newtoggle{submission}
\toggletrue{submission}
\iftoggle{submission}{
  
}{
  
}

\usepackage{pifont}% http://ctan.org/pkg/pifont
\usepackage{multicol,tabularx,capt-of}
\usepackage{multirow}

\begin{document}

\title{Validating Quantum State Preparation Programs} 

\titlerunning{Validating Quantum State Preparation Programs}
\authorrunning{L. Li, A. Sharma, Z. Tagba, S. Frett, A. Potanin}

\author{Liyi Li\inst{1} \orcidID{0000-0001-8184-0244} \and
Anshu Sharma\inst{2} \orcidID{0000-0002-8686-5835} \and
Zoukarneini Difaizi Tagba\inst{1} \and
Sean Frett\inst{1} \orcidID{0009-0004-1671-9578} \and
Alex Potanin\inst{3} \orcidID{0000-0002-4242-2725}
}

\institute{Iowa State University, USA \\
\email{\{liyili2,difaizi,fretts\}@iastate.edu} \and
The College of William and Mary, USA \\
\email{agsharma@wm.edu}
\and
Australian National University, Australia \\
\email{alex.potanin@anu.edu.au}}
\maketitle

\begin{abstract}
 One of the key steps in quantum algorithms is to prepare an initial quantum superposition state with distinct features.
These \emph{state preparation} algorithms are essential to the behavior of quantum algorithms, and complicated state preparation algorithms are difficult to program correctly and effectively.
We present QSV: a high-assurance framework implemented with the Rocq proof assistant, permitting the development of quantum state preparation programs and validating them to correctly reflect quantum program behaviors.
The key is to reduce the program correctness assurance for a program containing a quantum superposition state to that of the program state without superposition.
The reduction enables the development of \textit{an effective framework for validating quantum state preparation algorithm implementations on a classical computer} — a problem considered hard and without a clear solution until now.
We utilize the QuickChick property-based testing framework to validate state preparation programs.
We evaluated the effectiveness of our approach across 5 case studies implemented using QSV; these cases are not simulatable on current quantum simulators.
 \vspace*{-0.5em}
 \vspace*{-0.3em}
\end{abstract}

\section{Introduction}
\label{sec:intro}

Despite recent advances ~\cite{morphq_bugs,fuzz4all,10.1109/ASE51524.2021.9678798,fortunato,long:24,QDiff}, quantum program developers still lack tools to quickly validate program correctness---testing a program with many test inputs in a short time---and other properties when writing comprehensive programs ~\cite{gill2024quantumcomputingvisionchallenges,MattSwayne}.
Testing quantum programs directly on quantum hardware is problematic because running actual quantum computers is expensive, and the probabilistic nature of quantum computing means repeated trials may be necessary to validate correctness, driving up costs further.
Ideally, we should ensure that a program satisfies user specifications before running it on the hardware.
Unfortunately, such a framework might not exist for validating a quantum program for arbitrary properties since quantum programs are not classically simulatable by naive state-vector simulation without an exponential number of classical bits relative to qubits.

A quantum validation framework needs to satisfy three key design goals.

\begin{itemize}
\item Programmers can develop a quantum program based on a proper abstraction, with respect to high-level program properties, without worrying too much about low-level gates.

\item The framework contains a scalable and effective validator to quickly judge the correctness of a user-defined program, as well as other properties, based on certain types of quantum program patterns.

\item The validated program can be compiled into a quantum circuit for execution.
\end{itemize}

 \begin{wrapfigure}{l}{7cm}
\hspace*{-0.3em}
  \includegraphics[width=0.6\textwidth]{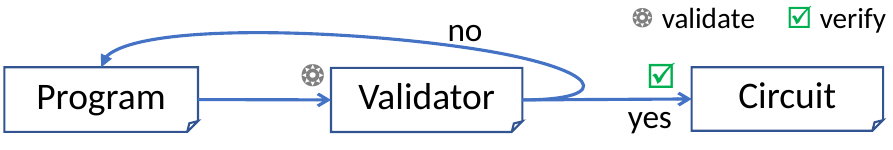}
   \caption{The QSV Flow}
\label{fig:qsv}
\end{wrapfigure}

We propose the Quantum State Preparation Program Validation Framework (QSV) (flow in \Cref{fig:qsv}), permitting effective validation of state preparation programs. \emph{The limitation is discussed in \Cref{sec:conclusion}}.
It includes three components. The first is the language, PQASM, which extends OQASM (Oracle Quantum Assembly Language); the P in PQASM stands for Preparation. To avoid bugs in common libraries \cite{morphq_bugs} and to avoid using operators outside cases they're specific to (sometimes, there is no documentation to tell what cases they allow), we chose to implement our own operators in PQASM. Users can develop their programs in the PQASM language, enabling them to write state-preparation programs at a high level of abstraction. Such programs can be validated by our validator (the second component), based on QuickChick \cite{quickchick} (a Rocq property-based testing facility), and we show several quantum program patterns effectively validated via our framework. 
Once a program is adequately developed in QSV, users can use our certified circuit compiler to compile it into a quantum circuit that runs on quantum hardware.

\myparagraph{Motivating Examples.}\label{sec:motivation}
Below is a simple state preparation subroutine, preparing a superposition of $n$ distinct basis-ket states, appearing in many algorithms \cite{Gorjan2007,mike-and-ike}. It has the following program transition property (a pre-state is transitioned to a post-state connected by $\to$) with program input of a length $m$ qubit array, initialized as $\aket{0}{m}$, and output a superposition of $n$ different basis-ket states, each with basis-vector $\aket{k}{m}$.

{\small
\begin{center}
$
\aket{0}{m}\to\sum_{j=0}^{n-1}\frac{1}{\sqrt{n}}\aket{j}{m}
$
\end{center}
}

\emph{A state preparation program can be defined as the starting component of a quantum algorithm,} typically starting with a length $m$ qubit array, each qubit initialized as zero ($\aket{0}{m}$) state, and preparing a superposition state $\sum_j \alpha_j \aket{c_j}{m}$ --- a linear sum of pairs (basis-kets) of complex amplitude $\alpha_j$ and bitstring (basis-vector) $c_j$ such that $\sum_j \slen{\alpha_j}^{2} =1 $ --- via a series of quantum operations.

Superposition is a key feature of quantum states, and quantum computers can execute programs with superposition states to query all possible inputs simultaneously, as discussed in \Cref{sec:background}.
Many quantum algorithms require a comprehensive design of state preparation components with different superposition structures, e.g., the $n$ basis-ket state in \Cref{fig:intros-example}.

A difficulty in developing these programs is that quantum program operations might affect every basis-ket in a superposition that might contain exponentially many basis-ket states, unlike the classical programs, where only a single basis-ket might be affected. In \Cref{fig:intros2}, a function $f$ is applied to every basis-ket in the quantum superposition state after all Hadamard operations were applied.
Moreover, many quantum languages are circuit-based \cite{Qiskit2019,tket,Cross2017,10.1145/3505636}.
These hinder the quantum program development as writing programs becomes unintuitive.

\Cref{fig:intros-example} shows a one-step procedure of the repeat-until-success program implementation for the $n$ basis-ket program.
The procedure starts with a series of Hadamard operations to prepare a uniform superposition of $2^m$ basis-kets as $\frac{1}{\sqrt{2^m}}\sum_{j=0}^{2^m}\aket{j}{m}$,
and then compares each basis-ket with the number $n$; such a comparison result is stored in the $y$ qubit.
If measuring $y$ results in $1$,  all the basis-kets ($\alpha_j\aket{j}{m}$), with bitstring numbers $j \ge n$, disappear, while those bitstrings $j<n$ will stay in $x$'s quantum state; thus, the correct state is prepared.
Since the measurement result $1$ is probabilistic, the repeat-until-success program requires repeating the one-step procedure many times to probabilistically prepare the target state.
In writing the program, a key component is the comparator $\qbool{x}{<}{n}{y}$, comparing every basis-vector of a quantum array $x$ with $n$ and storing the result in qubit $y$.
Such arithmetic operations have effective implementations \cite{oracleoopsla} and circuit-level optimizations \cite{VOQC,Xu2022}.
Therefore, QSV abstracts all these quantum arithmetic operations and relieves the pain of writing quantum programs.
The QSV compiler compiles and optimizes the arithmetic operations to quantum circuits. 
Moreover, we also provide types to classify different program patterns for users to write state preparation programs.

\begin{figure}[t]
{\hspace*{-1em}
\tiny
\begin{minipage}[t]{0.42\textwidth}
\subcaption{State Preparation Circuit}
\label{fig:intros2}
\vspace*{0.5em}
{\scriptsize
  \Qcircuit @C=0.5em @R=0.5em {
   & & \qw    & \gate{\cn{H}} & \qw & \multigate{3}{{ f(\ket{j})=(-i)^{j}\aket{j}{m}}}   & \qw &\qw & \qw \\
  \push{x:\aket{0}{m}\quad} & &  \vdots &          &     &                                          &    \rstick{\varphi_1} & &\\
   & & \vdots  &          &     &                                          &     &       &  \\
   & &  \qw   & \gate{\cn{H}} & \qw &  \ghost{{ f(\ket{j})=(-i)^{j}\aket{j}{m}}}         &\qw  &\qw    & \qw
      \gategroup{1}{2}{4}{2}{1em}{\{}
    }
}
\end{minipage}
\hfill
\begin{minipage}[t]{0.53\textwidth}
\subcaption{Preparing Superposition of $n$ Basis-kets}
\label{fig:intros-example}
\vspace*{0.5em}
 { \scriptsize
  \Qcircuit @C=0.5em @R=0.5em {
    &                     & & \qw & \gate{\cn{H}} & \qw & \qw & \multigate{6}{\texttt{$\qbool{x}{<}{n}{y}$}} & \qw & \qw & \qw & \\
    & \push{x:\aket{0}{m}\quad} & & \qw & \gate{\cn{H}} & \qw & \qw &  \ghost{\texttt{$\qbool{x}{<}{n}{y}$}}       &\qw & \qw & \qw &  \\
    & & &  &       & & &  & & & & \push{\varphi_2} \\
    & & &  & \dots & & & & & & & \\
    & & & & & & & & & & & \\
    &                     & & \qw & \gate{\cn{H}} & \qw & \qw &   \ghost{\texttt{$\qbool{x}{<}{n}{y}$}}      & \qw & \qw & \qw & \\
    & \push{y:\aket{0}{1}\quad} & & \qw & \qw      & \qw & \qw & \ghost{\texttt{$\qbool{x}{<}{n}{y}$}}        & \qw & \qw & \meter & \push{v} 
    \gategroup{1}{3}{6}{3}{1em}{\{}
    \gategroup{1}{11}{6}{11}{1em}{\}}
    }
  }
\end{minipage}
}
\caption{$x$ has $m$ qubits, and $y$ has $1$ qubit; $\varphi_1={ {\Msum_{j=0}^{2^m\sminus 1}(-i)^{j}\aket{j}{m}}}$. The $@$ symbol indicates in what qubit the result of an operation is stored. The right is one step in the repeat-until-success program to prepare the superposition state. $\varphi_2=\frac{1}{\sqrt{n}}\sum_{j=0}^n\aket{j}{m}$ if $v=1$; otherwise, $\varphi_2=\frac{1}{\sqrt{2^m-n}}\sum_{j=n}^{2^m}\aket{j}{m}$.}
\label{fig:circuits}
\end{figure}

Even if we permit high-level abstractions in QSV, validating a program might still be challenging because the number of basis-kets in a quantum state increases exponentially with the number of qubits, e.g., with $m$ qubits, the output state $\sum_{j=0}^{n-1}\frac{1}{\sqrt{n}}\aket{j}{m}$ might contain any of $2^m$ basis-kets and checking them individually might be challenging.
In developing our validator, we have two observations on quantum algorithms.
First, almost all quantum algorithms start with $m$ Hadamard operations to prepare a uniform superposition having $2^m$ basis-kets. These beginning Hadamard operations are simple enough and do not need to be validated, but they provide the source of a superposition state for later program operations to carve on. 
Second, even though quantum operations are probabilistic, one can "determinize" their behaviors.
If we consider the superposition of basis-kets as an array of basis-kets, quantum operations, except measurement, behave similarly to higher-order map functions applying to the quantum state. Measurement behaves similarly to a set selection, selecting a basis-ket element in the array, and the probability of such selection can be computed based on the amplitude value associated with the basis-ket.

QSV classifies beginning Hadamard operations as a special \cn{Had} type, to indicate that they are the source of the superposition state. When testing, instead of faithfully representing their operational behavior, QSV adopts a random pick of an individual basis-ket state as a representative, and validates programs based on transitions of the basis-ket. A measurement operation is then determinized and its probability is simply calculated via the amplitude value in the basis-ket.

For example, in dealing with the $n$ basis-ket program above, after applying the Hadamard operations, the program property is turned as the one on the left below, where $\sum_{j=0}^{2^m\sminus 1}\frac{1}{\sqrt{2^m}}\aket{j}{m}$ being the result of applying $m$ Hadamard operations. The QSV process of determinizing the basis-kets is to select a particular $j$, which turns the program property to be the right one.

{\footnotesize
\begin{mathpar}
 \inferrule[]{}{ \sum_{j=0}^{2^m\sminus 1}\frac{1}{\sqrt{2^m}}\aket{j}{m}\to\sum_{j=0}^{n\sminus 1}\frac{1}{\sqrt{n}}\aket{j}{m}}
  
   \inferrule[]{}{\forall j \in [0,2^m) \,,\, \frac{1}{\sqrt{2^m}}\aket{j}{m} \to(\frac{1}{\sqrt{n}}\aket{j}{m} \wedge j \in [0, n))}
    \end{mathpar}
}

Essentially, the validation process based on the right property is to assume a single basis-ket $\frac{1}{\sqrt{2^m}}\aket{j}{m}$ with a bitstring $j \in [0,2^m)$, then to validate to see if the output is the bitstring $j$ within range of $[0,n)$ and the associated amplitude being $\frac{1}{\sqrt{n}}$.
Via a property-based testing facility, such as QuickChick, by testing the program with enough candidate basis-kets, it will be highly likely that our validator can capture a bug if there is any. 
After validation, we compile the program to a quantum circuit via our certified compiler.

\vspace{0.2em}
\noindent\textbf{\textit{Contributions and Roadmap.}} 
We present QSV, enabling programmers to develop state-preparation programs.
Our contributions are as follows, with all Rocq proofs and experiment results available. Limitations are in \Cref{sec:conclusion}.

\begin{itemize}
\item We present the syntax, semantics, and type system of \pqasm, allowing users to define programs with a type-soundness proof in Rocq (\Cref{sec:pqasm}).

\item We develop a property-based testing (PBT) framework for validating programs written in \pqasm (\Cref{sec:implementation}) by showing a general flow of constructing such PBT frameworks for validating quantum programs.

\item We certify a compiler from \pqasm to \sqir \cite{VOQC} (\Cref{sec:vqir-compilation}) to ensure that our \pqasm tool correctly reflect quantum program behaviors.

\item We evaluate \pqasm via a selection of state preparation programs and demonstrate that QSV is capable of validating the programs (\Cref{sec:evaluation,sec:eval}),
which are hard to verify or validate if the input quibit length is normal (60 qubits per register and up to 361 qubits as the total input qubit size), based on the Qiskit quantum simulator and DDSim \cite{ddsim}.
\end{itemize}

                   % intro
\section{Background on Quantum Computing}
\label{sec:background}

\noindent\textbf{\textit{Quantum Data and Computation.}}
A quantum datum consists of one or more qubits. A single qubit can be represented as a two-dimensional vector $\begin{psmallmatrix} z_1 \\ z_2 \end{psmallmatrix}$, where $z_1$ and $z_2$ are complex amplitudes with $|z_1|^2 + |z_2|^2 = 1$. Using Dirac notation, this is written as $z_1\ket{0} + z_2\ket{1}$, with $\ket{0}$ and $\ket{1}$ as the \emph{computational basis-vectors}. When both $z_1$ and $z_2$ are non-zero, the qubit is in a superposition of the kets $\ket{0}$ and $\ket{1}$. The value inside the ket is a possible measurement outcome, and the norm square of the scalar factor, e.g., $\slen{z_1}^2$ and $\slen{z_2}^2$, is the probability of the measurement outcome.  Multi-qubit data is constructed via the tensor product, e.g., $\ket{0}\otimes\ket{1} = \ket{01}$. However, not all multi-qubit states can be separated into tensor products; some are entangled states, such as the Bell pair $\frac{1}{\sqrt{2}}(\ket{00}+\ket{11})$.

Quantum computation applies unitary gates to evolve the state of qubits. A gate is represented by a unitary matrix $U$ that acts on the qubit state vector $\ket{\psi}$, producing a new state $U\ket{\psi}$, e.g., the Hadamard gate $\cn{H} = \frac{1}{\sqrt{2}}\begin{psmallmatrix} 1 & 1 \\ 1 & -1 \end{psmallmatrix}$ transforms computational basis states into superpositions: applying $\cn{H}$ to $\ket{0}$ yields $\ket{+} = \frac{1}{\sqrt{2}}(\ket{0} + \ket{1})$, and applying $\cn{H}$ to $\ket{1}$ yields $\ket{-} = \frac{1}{\sqrt{2}}(\ket{0} - \ket{1})$. These operations are composed in quantum circuits for computations, where each wire denotes a qubit and each gate denotes a transformation applied at a specific time.

Measurement collapses the quantum state to a classical outcome with a probability determined by the amplitudes. For instance, measuring $\frac{1}{\sqrt{2}}(\ket{0}+\ket{1})$ yields $\ket{0}$ or $\ket{1}$, each with probability $\frac{1}{2}$. After measurement, the state irreversibly collapses to the observed basis state, and quantum coherence is lost.

\noindent\textbf{\textit{Quantum Oracles.}} Quantum algorithms manipulate input information encoded in ``oracles'', which are callable black-box circuits. 
Quantum oracles are usually quantum-reversible implementations of classical operations, especially arithmetic operations. Their behavior is defined in terms of transitions between single basis-kets.
We can infer the global state behavior based on the single basis-ket behavior through the quantum summation formula below. This resembles an array map operation in \Cref{fig:intros2}.
\oqasm in VQO~\cite{oracleoopsla} is a language that permits the definitions of quantum oracles with efficient verification and testing facilities by viewing quantum oracle operations as aggregate operations.

\noindent\textbf{\textit{Repeat-Until-Success Quantum Programs.}} A repeat-until-success program utilizes the probabilistic feature of partial measurement operations. It first sets up a one-step repeat-until-success by linking the desired quantum state with the success measurement of a certain classical value. If such a value is observed after measurement, the desired state is successfully prepared; otherwise, we repeat the one-step procedure.
One such example procedure is in \Cref{fig:intros-example} to repeat the $n$ basis-ket superposition state.
If we measure out $v=1$, the desired state $\varphi$ is prepared; otherwise, we repeat the procedure.

\noindent\textbf{\textit{No Cloning Theorem}} indicates there exists no general way of exactly copying a quantum datum \cite{noclone}. In quantum circuits, it relates to ensuring the reversibility of unitary gate applications.
For example, the controlled node and controlled body of a quantum control gate cannot refer to the same qubits, e.g., $\ictrl{q}{\iota}$ violates the property if $q$ is mentioned in $\iota$.
\pqasm{} enforces no cloning by typing.
              % background
\section{PQASM: A Language for Quantum State Preparations}
\label{sec:pqasm}

We designed \pqasm to express quantum-state preparation programs at a high-level abstraction.
\pqasm operations leverage a quantum-state design, with a type system to track the types of different qubits.
Such types restrict the kinds of quantum states, facilitating effective validation and analysis of the \pqasm program utilizing our quantum state representations.
This section presents \pqasm states and the language's syntax, semantics, typing, and soundness results.
More semantic and typing rules are in \Cref{appx:pqasm}.

As a running example, we program the $n$ basis-ket state preparation in \Cref{def:circuit-example} (figures in \Cref{fig:intros-example}).
The repeat-until-success program creates an qubit array $\overline{q}$ consisting of all zeroes and a new qubit $q'$, applies a Hadamard gate to each qubit in $\overline{q}$, uses a comparison operator $\qbool{\overline{q}}{<}{n}{q'}$ comparing $\overline{q}$ with $n$ and storing the result in $q'$, measures the qubit $q'$, and repeats the process until the measurment result is $1$. $\ihad{\overline{q}}$ is a syntactic sugar meaning $\iseq{\ihad{\overline{q}[0]}}{\iseq{...}{\ihad{\overline{q}[m\sminus 1]}}}$.

\begin{definition}[Example \pqasm program $P$ to prepare $n$ superposition states in $\overline{q}$]\label{def:circuit-example}\rm 
Qubit array length is at least $\cn{log}(n)+1$; $q'$ is a single qubit.

{\small
$
P\triangleq\iseq{\inew{\overline{q}}}{\iseq{\iseq{\inew{q'}}{\ihad{\overline{q}}}}{\iseq{\qbool{\overline{q}}{<}{n}{q'}}{\smea{x}{q'}{\sifb{x=1}{\sskip}{P}}}}}
$
}
\end{definition}

\subsection{\pqasm States and Syntax} \label{sec:pqasm-states}

\begin{figure}[t]
{\small
  \[\hspace*{-0.5em}
  \begin{array}{l}
  \begin{array}{l@{\;\;}c@{\;}c@{\;}l@{\qquad}l@{\;\;}c@{\;}c@{\;}l@{\;}@{\qquad}l@{\;\;}c@{\;}c@{\;}l@{\;}}
        \text{Qubit Name}  & q && & \text{Nat} & n,m & \in & \mathbb{N} &      \text{Real} & r & \in & \mathbb{R} \\[0.2em]
        \text{Complex} & z & \in & \mathbb{C} & \text{Bit} & b & \in & \{0,1\} & \text{Bitstring} & c & ::= & \overline{b}
  \end{array}\\
\begin{array}{l c c l@{\;}c@{\;}l@{\;}c@{\;}l}
      \text{Qubit Basis State} & \nu & ::= & \aket{b}{1} & & &\mid &\qket{r} \\
       \text{Qubit Records} & \theta & ::= & (\overline{q} &, & \overline{q} &, & \overline{q}) \\
      \text{Type} & \tau & ::= & \thadt & \mid& \tnort & \mid& \trott\\
      \text{Basis Vector} & \eta & ::= &  \Motimes_j\nu_j  \\
      \text{Basis-Ket} & \rho & ::= &  z\cdot \eta  \\
      \text{Quantum Data} & \varphi & ::= & \rho &\mid& \sum_{b=0}^1 \varphi \\
      \text{Quantum State} & \Phi & ::= & \theta \to \varphi
    \end{array}
    \end{array}
  \]
  \caption{State syntax. $\overline{S}$: a sequence of $S$. $\aket{c}{n\splus 1} \equiv \aket{c[0]}{1}\otimes...\otimes\aket{c[n]}{1}$, as $\slen{c}=n\splus 1$.}
  \label{fig:pqasm-state}
}
{\small 
\begin{center}
  $ \hspace*{-0.8em}
\begin{array}{l}
      \text{Classical Variable}~x,y \qquad\qquad \text{Boolean Expressions}~B\\
\begin{array}{llcl}
      \text{Parameters} & \alpha & ::= & \overline{q} \mid n\\
      \text{OQASM Arith Ops} & \mu & ::= & \iadd{\alpha}{\alpha} \mid \modmult{n}{\alpha}{m} \mid \qbool{\alpha}{=}{\alpha}{q} \mid \qbool{\alpha}{<}{\alpha}{q}\mid ...\\
      \text{Instruction} & \iota & ::= & \mu \mid \iry{r}{q} \mid \ictrl{q}{\iota} \mid \iseq{\iota}{\iota}\\
      \text{Program} & e & ::= & \iota \mid \iseq{e}{e} \mid \ihad{q} \mid \inew{{q}} \mid \smea{x}{\overline{q}}{e} \mid \sifb{B}{e}{e}
\end{array}
    \end{array}
  $
  \end{center}
}
  \caption{\pqasm syntax.}
  \label{fig:pqasm}
\end{figure}

A \pqasm program state $\Phi$ is represented based on the grammar in \Cref{fig:pqasm-state}, mapping from qubit records $\theta$ to a quantum datum $\varphi$.
A quantum datum is managed as qubit records, each of which is a collection of qubits possibly being entangled, while qubits in different records are guaranteed to have no entanglement.
Our quantum data consist of a quantum entanglement state that can be analyzed as two portions: 1) a sequence of sum operators $\sum_{b_1=0}^1...\sum_{b_n=0}^1$, and 2) a basis-ket $\rho$, a pair of a complex amplitude $z$ and a tensor product of basis vector $\eta$, which is a tensor of single qubit basis states $\nu$.
Each sum operator represents the creation of a superposition state via a Hadamard operation $\cn{H}$, i.e., the number of sum operators in a state represents the number of Hadamard operations applied to qubits in the state so far.
The variable $\rho=z\cdot \eta$ represents a basis-ket of a quantum state.
To understand the relation between a basis-ket and a quantum superposition state as a linear sum,
one can think of a superposition state as a collection of "quantum choices", and a basis-ket represents a possible choice, i.e., a measurement of a qubit record produces one possible choice, with the amplitude $z$ related to the probability of the choice.

A qubit basis state $\nu$ has one of two forms, $\aket{b}{1}$ and $\qket{r}$.
The former corresponds to the $\thadt$ and $\tnort$ types and the latter corresponds to the $\trott$ type.
The three types of qubit basis states are represented as the three fields in a qubit record,
i.e., $(\overline{q}_1, \overline{q}_2, \overline{q}_3)$ has three disjoint qubit sequences.
$\overline{q}_1$ is always typed as $\thadt$, $\overline{q}_2$ has type $\tnort$, and $\overline{q}_3$ has type $\trott$.
The $\thadt$ and $\tnort$ typed qubits are in the computational basis.
The $\trott$ typed basis state is different from the other types in terms of \emph{bases}, and it has the form $\qket{r} = \cn{cos}(r)\aket{0}{1}+\cn{sin}(r)\aket{1}{1}$, which is a basis state in the $Y$-rotated Hadamard basis.
Applying a $\cn{Ry}$ with the $Y$-axis angle $r$ to a $\aket{0}{1}$ qubit results in $\qket{r}=\cn{cos}(r)\ket{0}+\cn{sin}(r)\ket{1}$.
 
\Cref{fig:pqasm} presents \pqasm's syntax.  A \pqasm program $e$ is either an instruction $\iota$, a sequence operation $\iseq{e}{e}$, applying a Hadamard operation $\ihad{q}$ to a qubit $q$ to create a superposition, creating ($\inew{q}$) a new blank qubit $q$, a \cn{let} binding that measures a sequence of qubits $\overline{q}$ and uses the result $x$ in $e$, or classical conditional $\sifb{B}{e}{e}$ with classical Boolean guard $B$.
Each \cn{let} binding assigns a measurement result of qubits $\overline{q}$ to a variable $x$, representing a binary sequence.
%In measurement operations ($\mathpzc{M}$), we apply an operator to a \emph{qubit} $q$ or a sequence of quantum qubits $\overline{q}$.
We assume $\ihad{\overline{q}}$ and $\inew{\overline{q}}$ as syntactic sugars of applying a sequence of Hadamard and new-qubit operations.

The instructions $\iota$ correspond to unitary quantum circuit operations, including oracle arithemtic operations ($\mu$) implementable through \oqasm operations \cite{oracleoopsla} on a qubit sequence $\overline{q}$ (detailed in \Cref{sec:vqir}), a $Y$-axis rotation gate $\iry{r}{q}$ that rotates an angle $r$, a quantum control instruction ($\ictrl{q}{\instr}$), and a sequence operation ($\iseq{\iota}{\iota}$). Operation $\ictrl{q}{\instr}$ applies instruction $\instr$ \emph{controlled} on qubit $q$. 
In this paper, we provide several sample arithmetic oracle operations $\mu$ in \Cref{fig:pqasm}, such as addition ($\iadd{\alpha}{\alpha}$, adding the first to the second), modular multiplication ($\modmult{n}{\alpha}{m}$), quantum equality ($\qbool{\alpha}{=}{\alpha}{q}$), quantum comparison ($\qbool{\alpha}{<}{\alpha}{q}$), etc.
Each parameter $\alpha$ is either a group of qubits $\overline{q}$ or a number $n$.
Recall that a basis-ket state of a qubit array $\overline{q}$ is essentially a bitstring with a complex amplitude.
A quantum arithmetic operation applies the classical version of the operation to each basis-ket in a quantum superposition state, e.g., $\qbool{\overline{q}}{=}{n}{q}$ compares the bitstring representation of each basis-ket in $\overline{q}$ with the number $n$ and stores the result in $q$.
In a \pqasm program containing qubit array $\overline{q}$, $x$ in a \cn{let} binding binds a local classical value ($x$'s value) with the computational basis measurement result ($\mathpzc{M}$) on qubits $\overline{q}$. While the classical variable scope is local, the quantum qubits are immutable and globally scoped, i.e., quantum operations are applied to a global quantum state; each qubit in the state is referred to by quantum qubit names ($q$) in the program.
In \pqasm, we express a SKIP operation ($\sskip$) via a $\mu$ operation having empty qubits, as $
\mu \; \equiv\; \sskip\;\;\cn{when}\;\;FV(\mu)=\emptyset
$ ($FV$ collects free variables).

\subsection{Semantics}\label{sec:pqasm-dsem}

\begin{figure}[t]
{\footnotesize
\[
\begin{array}{lll}

\llbracket \mu \rrbracket\eta &= \app{\llbracket \mu \rrbracket\eta(\overline{q})}{\eta}{\overline{q}}
&
\texttt{where  }
FV(\mu) = \overline{q}
\\

\llbracket \iry{r}{q} \rrbracket\eta &=  \app{\qket{r}}{\eta}{q}
&
\texttt{where  }
\eta(q) = \ket{0}_{1}
\\

\llbracket \iry{r}{q} \rrbracket\eta &=  \app{\qket{\frac{3\pi}{2}- r}}{\eta}{q}
&
\texttt{where  }
\eta(q) = \ket{1}_{1}
\\

\llbracket \iry{r}{q} \rrbracket\eta &=  \app{\qket{r+r'}}{\eta}{q}
&
\texttt{where  }
\eta(q) = \qket{r'}
\\

\llbracket \ictrl{q}{\instr} \rrbracket\eta &=  \csem(\eta(q),\instr,\eta)
&
\texttt{where  }
\csem({\ket{0}{1}},{\instr},\eta)=\eta\quad\;\,
\csem({\ket{1}{1}},{\instr},\eta)=\llbracket \instr \rrbracket\eta
\\

\llbracket \iota_1; \iota_2 \rrbracket\eta &= \llbracket \iota_2 \rrbracket (\llbracket \iota_1 \rrbracket\eta)
\end{array}
\]
}
{\footnotesize
\begin{center}
$
\app{\eta'}{\eta}{\overline{q}}=\app{\eta'(q)}{\eta}{\forall q\in \overline{q}.\;q}
$
\end{center}
}
\caption{Instruction level \pqasm semantics; $\eta(\overline{q})$ : the states of the qubits $\overline{q}$ in $\eta$.}
  \label{fig:deno-sem}
\end{figure}

The \pqasm semantics has two levels: instruction and program levels.
The  former is a partial function $\llbracket - \rrbracket$ from an instruction $\instr$ and input basis vector state $\rho$ to an output state $\eta'$, written 
$\llbracket \instr \rrbracket\eta=\eta'$, shown in \Cref{fig:deno-sem}.
The \textit{program} level semantics is a labelled transition system $(\Phi,e) \xrightarrow{r} (\Phi',e')$ in \Cref{fig:exp-semantics}, stating that the input configuration $(\Phi,e)$ is possibly evaluated to an output configuration $(\Phi',e')$ with the probability $r$.
It essentially represents a Markov chain, where a program evaluation path is a chain of probabilities that shows the probability of reaching a particular configuration from the initial configuration.

The instruction level semantic rules assume that one can locate the state of a qubit $q$ in $\eta$ as $\eta(q)$, where we can refer to $\app{\nu}{\eta}{q}$ as updating the qubit state $\nu$ for the qubit $q$ in $\eta$.
Recall that a length $n$ basis vector state $\eta$ is a tuple of $n$ qubit values, modeling the tensor product $\nu_1\otimes \cdots \otimes \nu_n$. 
The rules implicitly map each qubit $q$ to a state, e.g., 
$\eta(q)$ corresponds to some sub-state $\nu_q$, where $\nu_q$ locates at the $q$'s position in $\eta$.
Many of the rules in \Cref{fig:deno-sem} update a \emph{portion} of a state. We write $\app{\nu_{q}}{\eta}{q}$ to update the state of the qubit of $q$ in $\eta$ with $\nu_q$, and
$\app{\eta'}{\eta}{\overline{q}}$ to update a range of qubits $\overline{q}$ according to the vector state $\eta'$, i.e., we update each $q\in\overline{q}$ with the qubit value $\eta'(q)$ and $\slen{\eta'}=\slen{\overline{q}}$.
The function \texttt{cu} is a conditional operation depending on the $\tnort$/$\thadt$ typed qubit $q$. 

\begin{figure*}[t]
{\scriptsize
  \begin{mathpar}
      \inferrule[S-Ins]{b = b_1,...,b_n \quad \varphi = \sum_{b_1=0}^1...\sum_{b_n=0}^1 z_b\cdot\eta_b \quad \llbracket \iota \rrbracket (\eta_b)=\eta'_b}{ (\Phi[\theta\mapsto \varphi],\iota) \xrightarrow{1}  (\Phi[\theta\mapsto \sum_{b_1=0}^1...\sum_{b_n=0}^1 z_b\cdot\eta'_b], \sskip)}
        
        \inferrule[S-New]{\overline{}}
        {(\Phi,\inew{q}) \xrightarrow{1} (\app{\aket{0}{1}}{\Phi}{(\emptyset,q,\emptyset)},\sskip)}
         \end{mathpar}
           \begin{mathpar}
        \inferrule[S-Had]{\overline{}}{ (\Phi[(\emptyset,q,\emptyset)\mapsto \aket{b}{1}],\ihad{q}) \xrightarrow{1} (\app{\sum_{j=0}^1 (\sminus 1)^{j\cdot b}\aket{j}{m}}{\Phi}{(q,\emptyset,\emptyset)}, \sskip) }
                   \end{mathpar}
           \begin{mathpar}
           \inferrule[S-Mea]{\Phi = \Phi' \uplus \{\uparrow\overline{q} \mapsto \Msum_j z_j\aket{c}{m}{\aket{c_j}{n}}+\qfun{\phi}{\overline{q}',c \neq \overline{q}'}\}  \\ r= \Msum_j \slen{z_j}^2 }
  {(\Phi,\smea{x}{\overline{q}}{e}) \xrightarrow{r} (\Phi'\uplus \{ (\uparrow\overline{q}) \textbackslash \overline{q} :  \Msum_j {\frac{z_j}{\sqrt{r}}}{\aket{c_j}{m}}\},e[c/x]) }
  \end{mathpar}
}
{\scriptsize
\begin{center}
$
\begin{array}{lcl}
\uparrow\overline{q}&\triangleq& \exists \overline{q}_1, \overline{q}_2,\overline{q}_3 \,.\,\uparrow\overline{q}=(\overline{q}_1, \overline{q}_2,\overline{q}_3)\wedge \overline{q} \subseteq \overline{q}_1 \uplus \overline{q}_2 \uplus \overline{q}_3
\\[0.2em]
\qfun{(\Msum_{i}{z_i}{\aket{c_{i}}{m}}{\aket{c'_i}{n}}+\varphi)}{\overline{q},b} &\triangleq& \Msum_{i}{z_i}{\aket{c_{i}}{m}}{\aket{c'_i}{n}}
      \qquad\texttt{where}\quad\forall i.\,\slen{c_{i}}=\slen{\overline{q}'}=m\wedge \denote{b[c_{i}/\overline{q}']}=\texttt{true}
\end{array}
$
\end{center}
}
\caption{Selected program Level \pqasm rules.}
\label{fig:exp-semantics}
\end{figure*}

\Cref{fig:exp-semantics} shows selected program-level semantic rules.
Since qubit records in a program state $\Phi$ partition the qubit domain, we can think of $\Phi$ as a multiset of pairs of qubit records and state values, as in \rulelab{S-Mea},
i.e., $\Phi[\theta \mapsto \varphi] \equiv \Phi \uplus \{\theta\mapsto \varphi\}$.
Rule \rulelab{S-Ins} connects the instruction level semantics with the program level by evaluating each basis vector state $\eta$ through the instruction $\iota$.
Rule \rulelab{S-New} creates a new blank ($\aket{0}{1}$) qubit, which is stored in the record $(\emptyset,q, \emptyset)$, a qubit ($q$) being created are $\tnort$ typed.
For a $\tnort$ typed qubit $(\emptyset, q, \emptyset)$, rule \rulelab{S-Had} turns the qubit to be $\thadt$ typed superposition, as $(q,\emptyset, \emptyset)$.
The measurement rule (\rulelab{S-Mea}) produces a probability $r$ label, and the value comes from the measurement result.
We first rewrite the quantum state to be a linear sum of computational basis-kets $\Msum_j r_j\aket{c}{m}{\aket{c_j}{n}}+\qfun{\varphi}{\overline{q}',c \neq \overline{q}'}$,
where every basis-vector $\aket{c}{m}$ (or $\aket{c_j}{n}$) is a bitstring, and all the sum operators are resolved as a single sum.

Any \pqasm state can be written as a sum of computational basis-kets.
As the equations shown below, the basis-ket state $\aket{c}{n}\qket{r}$ can be rewritten to be a sum of two computational basis-kets as $\cn{cos}(r)\aket{c}{n}\aket{0}{1}+\cn{sin}(r)\aket{c}{n}\aket{1}{1}$, while the two sum operators can be replaced as a single sum operator over length-$2$ bitstring $c$, where we replace $b_j$ with $c[0]$ (indexing $0$ of $c$) and $b_k$ with $c[1]$.

{\scriptsize
\begin{center}
$
\aket{c}{n}\qket{r} \equiv \cn{cos}(r)\aket{c}{n}\aket{0}{1}+\cn{sin}(r)\aket{c}{n}\aket{1}{1}
\qquad
\sum_{b_j=0}^1 \sum_{b_k=0}^1 \eta \equiv \sum_{c\in\{0,1\}^2} \eta[c[0]/b_j][c[1]/b_k]
$
\end{center}
}

\subsection{Typing}
\label{sec:pqasm-typing}

In \pqasm, typing is with respect to a \emph{type environment} $\Omega$, a set of qubit records partitioning qubits into different disjoint union regions, and a \emph{kind environment} $\Sigma$, a set tracking local variable scopes.
Typing judgments are two leveled, and are written as $\Omega\vdash \instr \triangleright_g \Omega'$ and $\Sigma; \Omega\vdash e\triangleright \Omega'$,
 which state that instruction $\instr$ and program expression $e$ are well-typed under $\Omega$ and $\Sigma$, and
transforms variables' bases as characterized by $\Omega'$.
$\Omega$ is populated via qubit creation operations (\cn{new}), while $\Sigma$ is populated via \cn{let} binding.
Selected typing rules are in \Cref{fig:exp-well-typed}.

The instruction level type system is flow-sensitive, where $g$ is the context flag and can be either $\cn{M}$ or $\cn{C}$, indicating whether the current instruction is inside a controlled operation. The program level type system communicates with the instruction level by assuming a \cn{C} mode context flag, shown in rule \rulelab{Tup}.
We explain the necessity of the context flag below.
Each qubit record $\theta$ represents an entanglement group, i.e., qubits in the same record might or might not be entangled, while qubits in different records are ensured not to be entangled.

{\footnotesize
\begin{center}
$
\begin{array}{l@{\;}c@{\;}l}
(\overline{q_1},\overline{q_2},\overline{q_3})\uplus(\overline{q_4},\overline{q_5},\overline{q_6}) &\equiv& 
(\overline{q_1}\uplus\overline{q_4},\overline{q_2}\uplus \overline{q_5},\overline{q_3}\uplus \overline{q_6})
\\[0.2em]
(\emptyset,\overline{q_1}\uplus \overline{q_2},\overline{q_3}\uplus \overline{q_4}) &\equiv& (\emptyset,\overline{q_1},\overline{q_3})\uplus (\emptyset,\overline{q_2},\overline{q_4})
\end{array}
$
\end{center}
}

\begin{figure}[t]
{\scriptsize\hspace*{-1em}
  \begin{mathpar}
      \inferrule[RyN]{\overline{}}{\{(\emptyset,\{q\},\emptyset)\} \uplus \Omega\vdash_{\cn{C}} \iry{r}{q}\triangleright \{(\emptyset,\emptyset,\{q\})\}\uplus \Omega}
\qquad 
     \inferrule[MuT]{\overline{q} \subseteq \overline{q_1}\cup \overline{q_2}}{\{(\overline{q_1},\overline{q_2},\overline{q_3})\}\uplus \Omega \vdash_g \mu(\overline{q})\triangleright \{(\overline{q_1},\overline{q_2},\overline{q_3})\}\uplus \Omega}
         \end{mathpar}
           \begin{mathpar}
     \inferrule[RyH]{\trot{\theta}=\{q\} \uplus \overline{q}}{\{\theta\}\uplus \Omega \vdash_{g}  \iry{r}{q}\triangleright \{\theta\}\uplus \Omega}
\quad
      \inferrule[Eqv]{\Omega\equiv \Omega' \\\\ \Omega' \vdash_g e\triangleright \Omega''}{ \Omega\vdash_g e\triangleright\Omega''}
\quad
      \inferrule[Tup]{\Omega\vdash_{\cn{C}} \instr\triangleright \Omega'}{\Sigma;\Omega \vdash \iota\triangleright \Omega'} 
\quad
      \inferrule[New]{q\not\in \Omega}{\Sigma;\Omega \vdash \inew{{q}}\triangleright \Omega \uplus \{(\emptyset,{q},\emptyset)\}} 
         \end{mathpar}
           \begin{mathpar}
    \inferrule[CuN]{\{\tnor{\theta}\downarrow \overline{q}\} \uplus \Omega \vdash_{\cn{M}} \iota \triangleright \{\tnor{\theta}\downarrow \overline{q}\}\uplus \Omega}{\{\tnor{\theta}\downarrow \{q\}\uplus\overline{q}\} \uplus \Omega \vdash_g \ictrl{q}{\iota} \triangleright \{\tnor{\theta}\downarrow \{q\}\uplus\overline{q}\} \uplus \Omega} 
\;
    \inferrule[HT]{\overline{}}{\Sigma;\Omega\uplus\{(\emptyset,q,\emptyset)\} \vdash \ihad{q}\triangleright \Omega\uplus\{(q,\emptyset,\emptyset)\}} 
         \end{mathpar}
           \begin{mathpar}
    \inferrule[CuH]{\{\thad{\theta}\downarrow \overline{q}\} \uplus \Omega \vdash_{\cn{M}} \iota \triangleright \{\thad{\theta}\downarrow \overline{q}\} \uplus \Omega}{\{\thad{\theta}\downarrow \{q\}\uplus\overline{q}\} \uplus \Omega \vdash_g \ictrl{q}{\iota} \triangleright \{\thad{\theta}\downarrow \{q\}\uplus\overline{q}\} \uplus \Omega} 
\quad            
      \inferrule[Mea]{\overline{q}\subseteq \theta \\ \Sigma\cup\{x\};\Omega \uplus\{\theta\textbackslash \overline{q}\} \vdash e \triangleright \Omega'}{\Sigma;\Omega \uplus\{\theta\} \vdash \smea{x}{\overline{q}}{e}\triangleright \Omega'} 
      \end{mathpar}
}
  \caption{Selected type rules. ${\footnotesize (\overline{q}_1,\overline{q}_2,\overline{q}_3) \backslash \overline{q} \triangleq (\overline{q}_1 \backslash\, \overline{q},\overline{q}_2 \backslash\, \overline{q},\overline{q}_3 \backslash\, \overline{q})}$; ${\footnotesize \overline{q}_1 \backslash\, \overline{q}}$: set subtract. }
  \label{fig:exp-well-typed}
\end{figure}

In our type system, we permit ordered equational rewrites among quantum qubit states.
Each type environment, mainly the operation $\uplus$, admits associativity, commutativity, and identity equational properties.
The $\uplus$ operations in the three fields in a qubit record also admit the three properties.
Other than the three equational properties, we admit the above partial order relations, where we permit the rewrites from left to right in our type system to permit qubit records merging and splitting.
Record merging can always happen, i.e., two qubit entanglement groups can be merged into one.
Qubit splitting cannot occur in $\thadt$ typed qubits.
A qubit record, including a $\thadt$ typed qubit, represents a quantum entanglement with qubits not separable,
while qubits being $\tnort$ and $\trott$ typed can be split into different records.
Rule \rulelab{Eqv} imposes the equivalence relation to permit the rewrites, only allowing rewrites from left to right, of equivalent qubit records.
Note that a type environment determines qubit record scopes in a quantum state $\varphi$ (\Cref{fig:pqasm-state}),
i.e., a quantum state should have the same qubit record domain as the type environment at a program point, so the equational rewrite of a type environment might affect the qubit state representation.

Other than the equational rewrites, the type system enforces three properties. 
First, it enforces that classical and quantum variables are properly scoped.
Rule \rulelab{Mea} includes the local variable $x$ in $\Sigma$.
Rule \rulelab{New} creates a new record $(\emptyset,q,\emptyset)$ in the post-environment, provided that $q$ does not appear in any record in $\Omega$.
In rule \rulelab{RyH}, the premise $\trot{\theta}=\{q\} \uplus \overline{q}$ utilizes $\trott$ to finds the $\trott$ filed in the record $\theta$ and ensures that $q$ is in the field. In rule \rulelab{CuH}, the premise $\thad{\theta}\downarrow \{q\}\uplus\overline{q}$ ensures that the controlled position $q$ is in the $\thadt$ field in $\theta$.
In rule \rulelab{MuT}, we ensure that the qubits $\overline{q}$ being applied by the $\mu$ operation are $\tnort$ and $\thadt$ typed, through the premise $\overline{q} \subseteq \overline{q_1}\cup \overline{q_2}$.

Second, we ensure that expressions and instructions are well-formed, i.e., any control qubit is distinct from the target(s), 
which enforces the quantum \emph{no-cloning rule}.
In rules \rulelab{CuH} and \rulelab{CuN} for control operations, when typing the target instruction $\iota$ (the upper level), we remove the control qubit $q$ in the records to ensure that $q$ cannot be mentioned in $\iota$.
In rule \rulelab{Mea}, we also remove the measured qubits $\overline{q}$ from the record $\theta$.

Third, the type system enforces that expressions and instructions leave affected qubits in a proper type ($\tnort$, $\thadt$, and $\trott$), representing certain forms of qubit states, mentioned in \Cref{fig:pqasm-state}; therefore, one can utilize the procedure mentioned in \Cref{sec:intro} to analyze \pqasm programs effectively.
The key is to utilize the summation formula to reduce the analysis of a general quantum state to that of a quantum state without entanglement.
Specifically, the $\iry{r}{q}$ opeartion is permitted only if $q$ is of $\tnort$ type, which is turned to $\trott$ type and stays there;
$\mu$ can be applied to $\tnort$ typed qubits $\overline{q}$ where $FV(\mu)=\overline{q}$; and a control qubit $q$ in $\ictrl{q}{\iota}$ can be applied to a $\tnort$ and $\thadt$ typed qubit.

\begin{wrapfigure}{r}{3cm}
{\small
$\begin{array}{c}
  \Qcircuit @C=0.5em @R=0.5em {
    &                     \qw & \ctrl{1} &  \qw \\
    &                     \gate{\cn{Ry}(0)}  & \gate{\cn{Ry}(r)}  & \qw    
    }
    \end{array}
$
}
\caption{Ensuring qubits inside a controlled \cn{Ry} have the same type.}
\label{fig:rygate}
\end{wrapfigure}

We ensure type restrictions for qubits via pre- and post-type environments.
Rule \rulelab{HT} permits the generation of $\thadt$ type qubits, a.k.a. superposition qubits $q$, provided that $q$ is of $\tnort$ type and not entangled with other qubits. Once a Hadamard operation is applied, we turn the qubit types to $\thadt$ in the post-type environment, so one cannot apply Hadamard operations again to the qubits.
This does not mean that users can only apply Hadamard operations once in \pqasm, because combining \cn{X} (our oracle operation) and \cn{Ry} gates can produce a Hadamard gate. We utilize the type information to locate the first appearance of Hadamard operations, identify them as the source of superposition, and apply treatments for them in our validation testing framework; see \Cref{sec:rand-testing}.

In Rules \rulelab{CuN} and \rulelab{CuH}, we ensure that the pre- and post-type environments are the same.
In $\ictrl{q}{\iota}$, if $\iota$ contains a \cn{Ry} operation, applying to a qubit $q$, $q$ must already be $\trott$ type.
\Cref{fig:rygate} provides a programming prototype satisfying this type requirement,
where programmers explicitly add a \cn{Ry} gate before the controlled \cn{Ry} operation to ensure the second qubit is in $\trott$ type.
The extra \cn{Ry} operation can be a $0$ rotation, equivalent to a SKIP operation, and can be removed by an optimizer when compiling to quantum circuits.
The above qubit type restriction does not depend on the applications of controlled operations.
We ensure this by associating the context flags with the instruction level type system.
When applying $\ictrl{q}{\iota}$, rules \rulelab{CuN} and \rulelab{CuH} turn the context flag to $\cn{M}$, indicating that $\iota$ lives inside a controlled operation.
Rule \rulelab{RyN} requires a context flag \cn{C}, meaning that the rule is valid only if the \cn{Ry} operation lives outside any controlled operation.
In contrast, rule \rulelab{RyH} does not require a specific context flag, which indicates that a \cn{Ry} operation inside a controlled node must apply to a qubit already in $\trott$ type.

\noindent\textbf{\textit{Soundness.}}
We prove type soundness: well-typed \pqasm programs are well-defined.
The type soundness theorem relies on a well-formed definition of a program $e$, $FV(e)\subseteq \Sigma$, meaning that all free variables in $e$ are bounded by $\Sigma$.
We also need the definition of the well-formedness of a \pqasm state as follows.

\begin{definition}[Well-formed state]\label{def:well-formed}\rm 
  A state $\Phi$ is \emph{well-formed},
  $\Omega \vdash \Phi$, iff:

\begin{itemize}
\item For every $q$ such that $\Omegaty(q) = \tnort$ or $\Omegaty(q) = \thadt$ ,
  $\Phi(q)$ has the form $\aket{b}{1}$.

\item For every $q$ such that $\Omegaty(q) = \trott$, $\Phi(q)$ has the form $\qket{r}$.
\end{itemize}
\end{definition}

\noindent
Type soundness is stated as two theorems: type progress and preservation; the proof is by induction on $\instr$ and is mechanized in Rocq.

\begin{theorem}\label{thm:type-progress}\rm[\pqasm Type Progress]
If $\emptyset; \Omega \vdash e \triangleright \Omega'$, $FV(e)\subseteq \emptyset$, and $\Omega \vdash \Phi$, then either $e = \sskip$ or there exists $r$, $e'$, and $\Phi'$, such that $(\Phi,e)\xrightarrow{r}(\Phi',e')$.
\end{theorem}

\begin{proof}
Fully mechanized proofs were done by induction on type rules using Rocq.
\end{proof}

\begin{theorem}\label{thm:type-sound-pqasm}\rm[\pqasm Type Preservation]
If $\Sigma; \Omega \vdash e \triangleright \Omega'$, $FV(e)\subseteq \Sigma$, $\Omega \vdash \Phi$ and $(\Phi,e)\xrightarrow{r} (\Phi',e')$, then there exists $\Omega_a$, such that $\Sigma;\Omega_a\vdash e' \triangleright \Omega'$ and $\Omega_a \vdash \Phi'$.
\end{theorem}

\begin{proof}
Fully mechanized proofs were done by induction on type rules using Rocq.
\end{proof}
                   % formal developement
\section{QSV Applications}
\label{sec:implementation}

This section presents QSV applications, based on \pqasm programs, to validate and compile the programs, including QSV's PBT (property-based testing) framework and the translation from \pqasm to \sqir and proof of its correctness. 

\subsection{Effectively Validating State Preparation Programs}\label{sec:rand-testing}

The QSV validator is built on PBT to ensure that a \pqasm program property is correct by attempting to falsify it using thousands of randomly generated input variables.
Case studies of our validator are in \Cref{sec:evaluation}.
Below, we show its construction.
We leverage the \pqasm's state representation and type system, which ensure that states can be represented effectively, to implement a validation framework for \pqasm programs using QuickChick \cite{quickchick}, a property-based testing (PBT) framework for Rocq in the style of Haskell's QuickCheck~\cite{10.1145/351240.351266}.
The framework has two main utilities: validating \pqasm program correctness properties and experimenting with effective implementations.
\pqasm measurement operations, due to the randomness inherent in quantum measurement, are difficult to test effectively, even with the assistance of program abstractions.
To have an effective validation framework, we restrict the properties that can be questioned to solely focus on the properties related to program correctness.

\noindent\textbf{\textit{Implementation.}}
PBT randomly generates inputs using a hand-crafted \emph{generator} and confirms that a property holds for these inputs. 
We develop a validator, based on the methodology in \Cref{sec:motivation}, using our effective (symbolic) state representation and carefully selecting the properties to validate for a program. 

\begin{figure}[h]
  \includegraphics[width=0.9\textwidth]{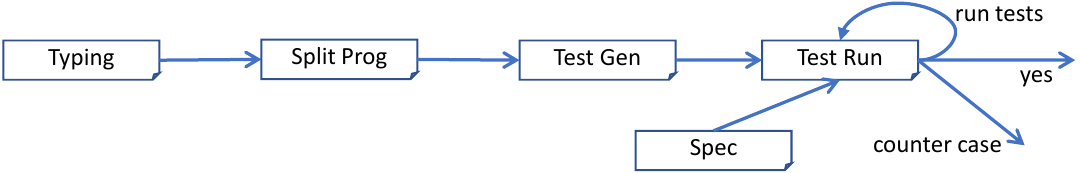}
   \caption{The Flow of PBT for QSV}
\label{fig:testing}
\end{figure}

\Cref{fig:testing} shows the flowof the PBT framework.
To validate a \pqasm program, we first utilize our \pqasm type system to generate a type environment $\Omega$ for qubits in a program $e$,
i.e., $\emptyset;\emptyset\vdash e \triangleright \Omega$, typing with an empty kind and type environment.
$\Omega$ partitions all qubits used in $e$ into three sets $(\overline{q_1}, \overline{q_2}, \overline{q_3})$, with $\overline{q_1}$, $\overline{q_2}$, and  $\overline{q_3}$ containing all qubits in $\thadt$ type, $\tnort$ type, and $\trott$ type, respectively.
In \pqasm, once a qubit is assigned to type $\thadt$, it stays in the type.
We utilize the property to locate all $\thadt$ typed qubits in a program $e$ and generate random testing data based on these qubits,
i.e., given the set ($\overline{q_1}$) of $\thadt$ typed qubits in $\Omega$, we generate random boolean values in $\{0,1\}$ for variables in $\overline{q_1}$.

A program is assumed in the form of $e=\iseq{\inew{\overline{q}}}{\iseq{\ihad{\overline{q}}}{e'}}$ \footnote{$\inew{\overline{q}}$ and $\ihad{\overline{q}}$ are syntactic sugar for multiple $\inew{q}$ and $\ihad{q}$ operations.},
i.e., all the $\cn{new}$ and $\cn{H}$ operations appear in the front of the program and $e'$ does not contain any such operations.
We then split the program by taking the $e'$ part and removing the $\cn{new}$ and $\cn{H}$ operations, assuming they have already been applied.
For some program patterns, such as repeat-until-success programs, we replace recursive process variables in $e'$ with $\sskip$ operations so that we only validate one step of a repeat-until-success. While these programs do not have a fixed termination time, QSV focuses on testing a single loop step, as we are not concerned with executing this program but performing property-based testing on it.
One example of programming splitting for \Cref{def:circuit-example} is given below; it removes the part $\iseq{\inew{\overline{q}}}{\iseq{\inew{q'}}{\ihad{\overline{q}}}}$ and replacing $P$ with $\sskip$.

{\small
\begin{center}
$
P'={\iseq{\qbool{\overline{q}}{<}{n}{q'}}{\smea{x}{q'}{\sifb{x=1}{\sskip}{\sskip}}}}
$
\end{center}
}

The ``Test Gen'' step in QSV (\Cref{fig:testing}) generates test cases to validate the key component $P'$ above.
In \Cref{def:circuit-example}, after applying the \cn{new} and \cn{H} operations, the post-state is $(\Phi,P')$, with $\Phi$ mapping $\theta$, entanglement groups, to superposition states $\varphi$. Each \cn{H} operation generates a single qubit uniformly distributed superposition state.
However, transitions over a superposition state make it hard to perform effective validation.
To resolve this, we treat quantum program operations $P'$ as higher-order map operations and validate the transition correctness based on basis kets, rather than over the whole superposition state.

Generally, given $\theta =(b_0,b_1,...,b_m,\overline{q_2},\overline{q_3})$, the superposition state $\varphi$ could also be written in the Dirac notation of $\sum_j \rho_j$ with $\rho_j = z_j \cdot \eta_j$, as follows.

{\small
\begin{center}
$
\sum_{b_0=0}^1 \sum_{b_1=0}^1 ... \sum_{b_m=0}^1 z(b_0,b_1,...,b_m)\cdot \eta({b_0,b_1,...,b_m})
$
\end{center}
}

Here, $b_0$, $b_1$, ..., $b_m$ are qubit variables assumed to have already been manipulated by the initial \cn{H} operations.
Applying a \cn{H} operation to a $\tnort$ typed qubit $\aket{b_a}{1}$ creates a uniformly distributed superposition $\sum_{b=0}^1 \frac{1}{\sqrt{2}} (-1)^{b\cdot b_a} \aket{b}{1}$, so it results in the state form above, with $z(b_0,b_1,...,b_m)$ and $\eta({b_0,b_1,...,b_m})$ being an amplitude and asis-vector formula, respectively.
We can then select the symbolic basis-ket state $z(b_0,b_1,...,b_m)\cdot \eta({b_0,b_1,...,b_m})$ as the representative basis-ket and rewrite $\varphi$ to be in the form $(\theta \to z\cdot \eta,P')$, with $b_0$, $b_1$, ..., $b_m$ acting as random variables.
For each variable $b_j$, we can randomly choose the value $b_j\in\{0,1\}$ for a particular test instance.
For $m$ qubits, we have $2^m$ different instances depending on the different value selection for the random variables.

For $P'$ above, we view each element in the $ m$-qubit array $\overline{q}$ as a random variable.
We generate an initial state $(\theta \to z(\overline{q}[0],...,\overline{q}[m\sminus 1])\cdot \eta(\overline{q}[0],...,\overline{q}[m\sminus 1]),P')$
with $\overline{q}[0],...,\overline{q}[m\sminus 1]$ being random variables and $\theta=(\overline{q},q,\emptyset)$.
We can then generate test instances for the random variables $\overline{q}[j]$ with $j\in[0,m)$.

As test instances are generated, we can validate the program by running each instance in our \pqasm interpreter, based on our program semantics.
The interpreter's result is provided as input to a specification checker to validate whether the specification is satisfied.
If the checker makes all test instances return \cn{true}, we validate the program; otherwise, we report a fault.
We show below the validation of different program properties.

\noindent\textbf{\textit{Validating Correctness.}}
We run a test instance in our interpreter with the initial state, as $(\theta \mapsto z\cdot \eta, P)\longrightarrow^*(\theta' \mapsto z'\cdot \eta',\sskip)$, where $\theta \equiv \theta'$. For a user-specified property $\psi$, a satisfiability check is applied to $\psi(z \cdot \eta, z' \cdot \eta')$ by replacing variables with $z \cdot \eta$ and  $z' \cdot \eta'$.
Recall the transformation of correctness property in \Cref{sec:intro} for the program in \Cref{def:circuit-example},
the key is to conduct validation on individual basis-kets rather than the whole quantum state.
Such property transformations can be summarized as the transformation from (1) to (3):

{\scriptsize
\begin{center}
$
\text{(1) }
\sum_j z_j \cdot \eta_j \to f(\sum_j z_j \cdot \eta_j)
\quad
\text{(2) }
\sum_j z_j \cdot \eta_j \to \sum_k g(z_k \cdot \eta_k) \wedge \phi(k)
\quad
\text{(3) }
\forall j. z_j \cdot \eta_j \to g'(z_j \cdot \eta_j) \wedge \phi(j)
$
\end{center}
}

The correctness property should be written in the format of (3).
The property (1) describes the program semantics, i.e., $f$ represents the semantic function for a program $e$.
The effects of such a function can always be in the form of a linear sum by moving the sum operator to the front, as in (2): one can always find $g$ such that $f(\sum_j r_j \cdot \eta_j) = \sum_k g(r_k \cdot \eta_k)$.
In some cases, we might need to insert the $\phi(k)$ predicate above, constraining index $k$ in the sum operator. 
Here, both $j$ and $k$ are indices for two different sum operators.
Often, the function $g$ can be turned into an equivalent form $g'$ based on the index $j$.
Note that for certain states, there are properties that hold for the state at large but do not imply that they hold for each term. For instance, when $k$ is a bitstring and $\cn{not}(k)$ performs the $\cn{not}$ operator on each value in the bitstring, e.g. $\cn{not}(011) = 100$, if $\sum_{i} \ket{k_i} = \sum_{i} \ket{\cn{not}(k_i)} $ that does not imply that $\ket{k_i}=\ket{\cn{not}(k_i)}$.

Once rewriting (1) to (2), we can then transform the formula to (3) without any sum operators, via the axiom of extensionality.
Such transformation might come with the index restriction $\phi$ based on index $j$, as shown in (3), meaning that we find a representative basis-ket state $r_j \cdot \eta_j$ for the input superposition state and output a basis-ket state $g'(r_j \cdot \eta_j)$, with the index restriction $\phi(j)$.

To validate $P'$ above, we transform the correctness property from the left to the right one below ($\Rightarrow$ is logic implication).

{\scriptsize
\begin{center}
$
x = 1 \Rightarrow \sum_j^{2^m}\frac{1}{\sqrt{2^m}}\aket{j}{m}\aket{0}{1} \to \sum_j^{n}\frac{1}{\sqrt{n}}\aket{j}{m}
\qquad
\forall j\in[0,2^m)\,.\, x = 1 \Rightarrow \aket{j}{m}\aket{0}{1} \to \aket{j}{m} \wedge j < n
$
\end{center}
}

The right property states that, if the measurement results in $1$ ($x = 1$), each basis-ket in the pre-state for $\overline{q}$ and $q$ (values $\aket{j}{m}$ and $\aket{0}{1}$) results in $\overline{q}$ being the same $\aket{j}{m}$ with the restriction $j < n$.
In implementing a validation property, the flags $\textcolor{spec}{m}$ and $\textcolor{spec}{1}$ essentially indicate the length of bitstring pieces, cast into a natural number for comparison.
To validate the correctness property against $P'$, we create a length $m$ bitstring for $\aket{j}{m}$ and view $j[k]$, for $k\in[0,m)$, being the $k$-th bit in the bitstring.
Recall that $P'$ has an initial basis vector state pattern as $\aket{\overline{q}[0]}{1}...\aket{\overline{q}[m\sminus 1]}{1}\aket{0}{1}$.
Here, $\aket{0}{1}$ is the state for qubit $q$ and $\overline{q}[0],...,\overline{q}[m\sminus 1]$ are random variables for qubit array $\overline{q}$.
To check the property, we bind each $j[k]$ with $\overline{q}[k]$ for $k\in[0,m)$, and see if the output basis-ket state results in the same $\overline{q}[0],...,\overline{q}[m\sminus 1]$. We also check if $\overline{q}$'s natural number representation is less than $n$,
i.e., we turn $\overline{q}[0],...,\overline{q}[m\sminus 1]$ to a number and compare it with $n$.

\noindent\textbf{\textit{Validating Other Properties.}}
The above procedure is only useful in validating correctness.
There might be other interesting properties, such as probability and effectiveness properties.
For example, in validating \Cref{def:circuit-example}, we might want to ask how likely the qualified state can be prepared,
which is hard to validate in general, but it can be effectively sampled out in some cases.

{\small
\begin{center}
$
x = 1 \Rightarrow \sum_j^{2^m}\frac{1}{\sqrt{2^m}}\aket{j}{m}\aket{0}{1} \to \sum_j^{n}\frac{1}{\sqrt{n}}\aket{j}{m}
$
\end{center}
}

In the property for $P'$ above, the basis-vector number $n$ is related to $\overline{q}$'s qubit number $m$, as $n\in[0,2^m)$.
We explain how to validate the program's effectiveness, i.e., the probability that the repeat-until-success program produces the correct state.
Note that superposition states are always uniformly distributed without any \cn{Ry} operations.
The success rate of preparing a superposition state in the repeat-until-success scheme is the ratio between the number of desired basis-vector values ($< n$) and the total number of possible basis-vector values, i.e., different combinations of $\overline{q}[0],...,\overline{q}[m\sminus 1]$. 
We can validate the effectiveness by dividing the number of different desired basis-vectors and the number of possible values in $\overline{q}$; that is, $2^m$.
In general, assume that we have a basis-vector expression $e(\overline{q})$ for $m$ qubits $\overline{q}$, a measurement statement $\mathpzc{M}\cn{(}e(\overline{q})\cn{)}$ storing the result in $v$, and have a boolean check on $v$ as $B(v)$ defining the good states.
By assigning $\{0,1\}^m$ for $\overline{q}$, the probability of having the good states is the division of number of good state values ($B(e(\overline{q}))=\cn{true}$) and the possible basis-vector numbers in $e(\overline{q})$ for all possible assignments.
Such a property can be validated by sampling.
In some complicated cases, the right property might be hard to validate, but one can always use Rocq to verify the effectiveness via the above scheme.

\noindent\textbf{\textit{Performance Optimizations.}}
We took several steps to improve validation performance, e.g., we streamlined the representation of states: per the semantics in \Cref{fig:deno-sem}, in a state with $n$ qubits, the amplitude associated with each qubit can be written as $\Delta(\frac{\upsilon}{2^n})$ for some natural number $\upsilon$. 
Qubit values in both bases are thus pairs of natural numbers: the global phase $\upsilon$ (in range $[0,2^n)$) and $b$ (for $\aket{b}{1}$) or $y$ (for $\qket{\frac{y}{2^n}}$). 
A \pqasm state $\varphi$ is a map from qubit positions $p$ to qubit values $q$; in our proofs, this map is implemented as a partial function, but for validation, we use an AVL tree (a kind of self-balancing binary tree) implementation (proved equivalent to the functional map). 
To avoid excessive stack use, we implemented the \pqasm semantics function tail-recursively. 
To run the tests, QuickChick runs OCaml code that it \emph{extracts} from the Rocq definitions; during extraction, we replace natural numbers and operations thereon with machine integers and operations. Performance results are in \Cref{sec:evaluation}.

\subsection{Translation from \pqasm to \sqir}\label{sec:vqir-compilation}

\newcommand{\tget}{\texttt{get}}
\newcommand{\tstart}{\texttt{start}}
\newcommand{\tfst}{\texttt{fst}}
\newcommand{\tsnd}{\texttt{snd}}
\newcommand{\tucom}[1]{\texttt{ucom}~{#1}}
\newcommand{\tif}{\texttt{if}}
\newcommand{\tthen}{\texttt{then}}
\newcommand{\telse}{\texttt{else}}
\newcommand{\tlet}{\texttt{let}}
\newcommand{\tin}{\texttt{in}}

We translate \pqasm to \sqir by mapping \pqasm virtual qubits to \sqir \cite{VOQC},  a quantum circuit language based on Rocq, concrete qubit indices, and expanding \pqasm instructions to sequences of \sqir gates.
To express the classical components of quantum algorithms, \sqir typically utilizes Rocq program constructs.
To define our compiler, we use the \sqir one-step nondeterministic semantics, which includes one-step operational semantics for simple Rocq constructs, such as conditionals and classical sequential operations.

Our translation is expressed as the two-level judgments
$\Xi\vdash \iota \gg \epsilon$ and $(n,\Xi, e) \gg (n',\Xi',\chi)$, where $\epsilon$ is the output \sqir circuit, and $\Xi$ and $\Xi'$ map a \pqasm qubit $q$ to a \sqir concrete qubit index (i.e., offset into a  global qubit register), $\chi$ is a hybrid program including \sqir quantum circuits and Rocq classical programs,
and $n$ and $n'$ are the qubit sizes in the whole system.

{\scriptsize
  \begin{mathpar}
    \inferrule[CRy]{\overline{}}{\Xi \vdash \iry{r}{q} \gg (\gamma,\textcolor{blue}{\iry{r}{(\Xi(q))}})}
    \quad
    \inferrule[CCU]{\Xi\vdash \instr \gg \textcolor{blue}{\epsilon}\quad
      \textcolor{blue}{\epsilon' = \texttt{ctrl}(\gamma(q),\epsilon)}}{\Xi\vdash\ictrl{p}{\instr} \gg \textcolor{blue}{\epsilon'}}    
         \quad               
    \inferrule[CNew]{\Xi'=\Xi[\forall q\in\overline{q}\,.\,q \mapsto \slen{\Xi}+\cn{ind}(\overline{q},q)]}{(n,\Xi, \inew{\overline{q}}) \gg (n+\slen{\overline{q}},\Xi', \sskip)}             
  \end{mathpar}
}

We show selected rules above; more are in \Cref{sec:vqir}.
Rules \rulelab{CRy} and \rulelab{CCU} are the instruction-level translation rules, which translate a \pqasm instruction to a \sqir unitary operation.
$\iry{r}{q}$ has a directly corresponding gate in \sqir.
In the \texttt{CU} translation, the rule assumes that $\instr$'s translation does not affect the $\Xi$ position map. This requirement is assured for well-typed programs per rule \rulelab{CU} in \Cref{fig:exp-well-typed}. 
 \texttt{ctrl} generates the controlled version of an arbitrary \sqir program using standard decompositions \cite[Chapter 4.3]{mike-and-ike}.

Rule \rulelab{CNew} is one of the program-level rules, translating a qubit creation operation in \pqasm.
In \sqir, there is no qubit creation, in the sense that every qubit is assumed to exist in the first place.
The translation essentially translates the operation to a SKIP operation in \sqir and increments the qubit heap size in the generated \sqir program.

\newcommand{\transs}[3]{[\!|{#1}|\!]^{#2}_{#3}}

Below is the translation correctness theorem.
The proof utilizes the \sqir nondeterministic semantics, where a qubit measurement produces two possible outcomes with different probabilities associated with them, 
i.e., this nondeterministic semantics is essentially \sqir's way of describing a Markov-chain procedure.
To formally state the correctness property, we relate \pqasm superposition states $\Phi$ to \sqir states, as $\denote{\Phi}^{n'}$, which are vectors of $2^{n'}$ complex numbers.
We can utilize $\Xi$ to relate qubits in \pqasm with qubit positions in \sqir.

\begin{theorem}\label{thm:vqir-compile}\rm[Translation Correctness]
  Suppose $\Sigma; \Omega \vdash e \triangleright \Omega'$ and
  $(n,\Xi,e) \gg (n',\Xi',\chi)$.
Then for $\Omega \vdash \Phi$ and $(\Phi,e)\xrightarrow{r}(\Phi',e')$, we have $(\denote{\Phi}^{n'},\chi)\xrightarrow{r}(\denote{\Phi'}^{n'},\chi')$ and $(n',\Xi',e')\gg (n'',\Xi'',\chi')$.
\end{theorem}

\begin{proof}
The proof is by induction on the \pqasm program $e$. 
Most of the proof simply shows the correspondence of operations in $e$ to their translated-to gates $\epsilon$ in \sqir, except for \cn{new} and measurement operations, which update the $\Xi$ map.
\end{proof}

\section{Evaluation: Applicativity via Case Studies}\label{sec:evaluation}

Here, an experimental evaluation of QSV across many programs is conducted to assess how QSV can be used to effectively build and validate useful quantum state preparation programs with different patterns. The efficiency and scalability comparison with other frameworks are in \Cref{sec:eval}.

\noindent\textbf{\textit{Implementation.}} We implement QSV in Rocq and utilize QuickChick to create a quantum program validation framework.
We also provide a \pqasm circuit compiler that translates \pqasm programs to OpenQASM via SQIR.

\noindent\textbf{\textit{Experimental Setup.}} We perform our evaluation on an Ubuntu computer, which has 8-core 13-gen i9 Intel processors and 16 GiB DDR5 memory.

\begin{figure}[t]
{\scriptsize
\begin{center}
\begin{tabular}{| l | c | c | c | c |}
\hline
 State Preparation Program  & 8B Gate \#  & 8B Qubit \# & 60B Gate \#  & 60B Qubit \# \\
 \hline
$n$ basis-ket & 766 & 9 & 5,732 & 61\\
Modular Exponentiation & 277K & 27 & 3.3M & 183\\
Amplitude Amplification & 65 & 9 & 481 & 61 \\
Hamming Weight & 9,088 & 16 & 170K & 120 \\
Distinct Element & 16,036  & 49 & 120K &  361 \\
\hline                           
\end{tabular}
\end{center}
}
\caption{Program statistics; a single register with 8 and 60 Qubits (8B/60B); Distinct elements have 5 distinct elements (keys). ‘K’: thousand, ‘M’: million.}
\label{fig:qiskit-data}
\end{figure}

We present several case studies here to demonstrate the construction and validation of state preparation programs.
We classify \emph{two different program patterns} below and examine their validation strategies.
We list the qubit and gate counts for these programs in \Cref{fig:qiskit-data}.
The data are collected by compiling our programs to SQIR using the elementary gateset $\{\cn{X}, \cn{H}, \cn{CX}, \cn{Rz}\}$.
To describe the programs, we define the following repeat operator for repeating a process $n$ times.
The function $P$ takes a natural number and outputs a quantum program.

{\small
\begin{center}
$
Re(P,n)\triangleq\sifb{n=0}{\sskip}{\iseq{Re(P,n-1)}{P(n-1)}}
$
\end{center}
}

\subsubsection{Quantum Loop Programs.}\label{sec:quantumloop}
We first examine the class of programs only involving quantum unitary gates without measurement.
Such programs typically contain quantum loops, a repetition of subroutines formed via unitary gates.
These programs usually act as a large part of some quantum algorithms.
For example, the modular exponentiation state preparation program is the major part of Shor's algorithm, while the amplitude amplification state preparation program is the major part of an upgraded amplitude estimation algorithm \cite{10.1007/s11128-019-2565-2}.

\noindent\textbf{\textit{Modular Exponentiation State Preparation.}}\label{sec:modmult}
In Shor's algorithm \cite{shors}, we prepare the the superposition state of modular exponentiation based on two natural numbers $c$ and $n$ with $\cn{gcd}(c,n)=1$, as $\varphi_3 = \frac{1}{\sqrt{2^m}} \sum_j^{2^m}\aket{j}{m}\aket{\modexp{c}{j}{n}}{m}$.

{\footnotesize
\begin{center}
$
\begin{array}{l@{\;}c@{\;}l}
Q(\overline{q_1},\overline{q_2})(k) &\triangleq& \ictrl{(\overline{q_1}[k])}{\modmult{c^{2^k}}{\overline{q_2}}{n}}
\\
P(m) &\triangleq& \iseq{\inew{\overline{q_1}}}{\iseq{\iseq{\inew{\overline{q_2}}}{\ihad{\overline{q_1}}}}{{Re(Q(\overline{q_1},\overline{q_2}),m)}}}
\end{array}
$
\end{center}
}

The program to prepare the modular exponentiation state is listed above, with the circuit diagram in \Cref{fig:mod-mult}.
The program starts with two new length $m$ qubit arrays $\overline{q_1}$ and $\overline{q_2}$, and turns $\overline{q_1}$ to a uniform superposition by $m$ \cn{H} gates. We then repeat $m$ times a controlled modular multiplication ($\ictrl{(\overline{q_1}[k])}{\modmult{c^{2^k}}{\overline{q_2}}{n}}$) application, controlling on the qubit $\overline{q_1}[k]$ with application to $\overline{q_2}$.

As in \Cref{sec:rand-testing}, to conduct the the program correctness validation, we cut off the first three operations ($\texttt{new}$ and $\cn{H}$ operations) in $P(m)$, producing a program piece $Re(Q(\overline{q_1},\overline{q_2}),m)$, where $Q(\overline{q_1},\overline{q_2})$ is a function taking in a number $k$ and outputting a program $\ictrl{(\overline{q_1}[k])}{\modmult{c^{2^k}}{\overline{q_2}}{n}}$.
The correctness specification is also transformed without the superposition state description below.

{\small
\begin{center}
$
\forall j \in [0,2^m)\,.\,\aket{j}{m}\aket{0}{m}\to\aket{j}{m}\aket{\modexp{c}{j}{n}}{m}
$
\end{center}
}

Validating the loop program $Re(Q(\overline{q_1},\overline{q_2}),m)$ essentially executes $Q$ $m$ times.
In executing the $k$-th loop step, we have the following loop invariant.

{\small
\begin{center}
$
\aket{j}{k}\aket{\modexp{c}{j}{n}}{m}\to\aket{j}{k\splus 1}\aket{\modexp{c}{j}{n}}{m}
$
\end{center}
}

In the pre-state, $\aket{j}{k}$ is a length $k$ bitstring while the post-state has $\aket{j}{k\splus 1}$ being length $k\splus 1$.
To understand the behavior, notice that we have the most significant bit on the right. 
A $k\splus 1$ length bitstring $\aket{j}{k\splus 1}$ can be expressed as a composition over a length $k$ bitstring as $\aket{j}{k}\aket{0}{1}$ or $\aket{j}{k}\aket{1}{1}$, with the most sigificant bit being $0$ or $1$. 
For the former case, applying the controlled modular multiplication results in $\aket{j}{k}\aket{0}{1}\aket{\modexp{c}{j}{n}}{m}$,
i.e., the $\overline{q_2}$ part of the bitstring remains the same.
For the latter case ($\aket{j}{k}\aket{1}{1}$), the controlled modular multipliation results in the state $\aket{j}{k}\aket{1}{1}\aket{\modmult{c^{2^{k}}}{\modexp{c}{j}{n}}{n}}{m}=\aket{j}{k}\aket{1}{1}\aket{\modexp{c}{j+2^k}{n}}{m}$.
Both cases can be rewritten to the post-state in the loop invariant above.

\begin{figure}[t]
{\hspace*{-1em}
\begin{minipage}[b]{0.42\textwidth}
  {\tiny
  $
  \Qcircuit @C=0.5em @R=0.5em {
    &                     & & \gate{\cn{H}} & \ctrl{6} & \qw & \qw & \qw & \qw & \qw & \qw & \\
    & \push{\overline{q_1}:\ket{0}\;\,} & & \gate{\cn{H}} & \qw & \ctrl{5} & \qw & \qw & \qw & \qw & \qw &  \\
    &                     & &  & & & &  & & & & \\
    &                     & & \dots & & & & \dots & & & & \\
    &                     & & & & & & & & & & \\
    &                     & & \gate{\cn{H}} & \qw & \qw & \qw & \qw & \qw & \ctrl{1} & \qw & \push{\;\;\varphi_3}\\
    &                     & & \qw & \multigate{4}{\texttt{$A({2^0})$}} & \multigate{4}{\texttt{$A({2^1})$}} & \qw & \qw & \qw & \multigate{4}{\texttt{$A(2^{m-1})$}} & \qw & \\
    & \push{\overline{q_2}:\ket{1} \;\,} & & \qw & \ghost{\texttt{$A({2^0})$}} & \ghost{\texttt{$A({2^1})$}} & \qw & \qw & \qw & \ghost{\texttt{$A(2^{m-1})$}} & \qw &  \\
    & & & \dots & & & & \dots & & & & \\
    & & & & & & & & & & & \\
    & & & \qw & \ghost{\texttt{$A({2^0})$}} & \ghost{\texttt{$A({2^1})$}} & \qw & \qw & \qw & \ghost{\texttt{$A(2^{m\sminus 1})$}} & \qw
    \gategroup{1}{3}{5}{3}{1em}{\{}
    \gategroup{7}{3}{11}{3}{1em}{\{}
    \gategroup{1}{11}{11}{11}{1em}{\}}
    }
$
    }
    \caption{Modular Exponentiation Circuits; $A(v)=(c^{v} * \overline{q_2}) \,\cn{\%}\, n$.}
\label{fig:mod-mult}
\end{minipage}
\hfill
\begin{minipage}[b]{0.51\textwidth}
{\tiny
$
  \Qcircuit @C=0.5em @R=0.5em {
    &                     & & \gate{\cn{H}} & \ctrl{6} & \qw & \qw & \qw & \qw & \qw & \qw & \qw & \qw & \qw &  \\
    &                     & & \gate{\cn{H}} & \qw & \ctrl{5} & \qw & \qw & \qw & \qw & \qw & \qw & \qw & \qw & \\
    & \push{\overline{q}:\ket{0}\;} & &  & & & &  & & & & & & & \\
    &                     & & \dots & & & & \dots & & &  & & & & \push{\;\;\varphi_4}\\
    &                     & & & & & & & & & & & & & \\
    &                     & & \gate{\cn{H}} & \qw & \qw & \qw & \qw & \qw & \qw & \qw & \qw & \ctrl{1} & \qw & \\
    & \push{q':\ket{0}\;} & & \gate{B(2^n)}  & \gate{B(2^{n\sminus 1})}      & \gate{B(2^{n\sminus 2})}  & \qw & & \dots  & & & \qw  & \gate{B(2^{0})} & \qw 
    \gategroup{1}{3}{6}{3}{1em}{\{}
    \gategroup{1}{14}{7}{14}{1em}{\}}
    }
$
}
\caption{Amplitude amplification state preparation; $B(v)=\cn{Ry}(\frac{r}{v})$.}
\label{fig:aacircuit}
\end{minipage}
}
\end{figure}

To validate the program piece $Re(Q(\overline{q_1},\overline{q_2}),m)$ against the above specification,
we are given a length $2m$ bitstring with $\overline{q_1}$ and $\overline{q_2}$ both being length $m$,
and view $\overline{q_1}$ as a length $m$ array of random variables, and prepare a initial state $\aket{\overline{q_1}[0],...,\overline{q_1}[m\sminus 1]}{m}\aket{0}{m}$, with random generation of length $m$ binary bitstrings, as test instances, for the variables $\overline{q_1}[0],...,\overline{q_1}[m\sminus 1]$.
We then use the mechanism in \Cref{sec:rand-testing} to validate the program piece.
As we mentioned in \Cref{sec:intro}, after we test enough samples of different input basis-ket states by picking different values for random variables, we have a high assurance that the modular exponentiation state preparation program correctly prepares a superposition state.

Note that the program is a deterministic quantum circuit program, so the probability of preparing the superposition state is $100\%$.  

\noindent\textbf{\textit{Amplitude Amplification State Preparation Through Ry Gates.}}\label{sec:aary}
In the amplitude amplification algorithm, one needs to prepare a special superposition state \cite{10.1007/s11128-019-2565-2}, as $\varphi_4=\frac{1}{\sqrt{2^n}}\sum_{j=1}^{2^n}\aket{j}{n}\qket{\frac{(2j+1)r}{2^n}}$, circuit in \Cref{fig:aacircuit}.
Then, the amplitude amplification algorithm utilizes the last qubits ($\qket{\frac{(2j+1)r}{2^n}}$) to amplify the amplitudes of the basis-kets having a particular property with respect to some $j$. The $r$ value is the upper limit of the possible amplitude value, i.e., we want to carefully select $r$ to ensure $\frac{(2j+1)r}{2^n}\in[0,\frac{\pi}{2})$.

{\footnotesize
\begin{center}
$
\begin{array}{l@{\;}c@{\;}l}
Q(\overline{q},q')(j) &\triangleq& \ictrl{(\overline{q}[j])}{\iry{\frac{r}{2^{n-j}}}{q'}}
\\[0.2em]
P(n) &\triangleq& \iseq{\inew{\overline{q}}}{\iseq{\iseq{\inew{q'}}{\iseq{\ihad{\overline{q}}}{\iry{\frac{r}{2^n}}{q'}}}}{{Re(Q(\overline{q},q'),n)}}}
\end{array}
$
\end{center}
}

We implement the program $P$ in \pqasm with the input of a qubit array $\overline{q}$ and a single qubit $q'$.
We then apply \cn{H} gates on $\overline{q}$ and a $Y$-axis rotation on $q'$.
Eventually,we apply a series of controlled $Y$-axis rotation operations --- controlling on the $\overline{q}[j]$ qubit ($j\in[0,n)$) and applying $\cn{Ry}$ on $q'$;
each single controlled $Y$ axis roation is handled by the $Q$ process.

Since there is no measurement, the success rate of preparing the amplitude amplification state is theoretically $100\%$.
We mainly test the correctness here.

{\small
\begin{center}
$
{\iseq{{{\iry{\frac{r}{2^n}}{q'}}}}{{Re(Q(\overline{q},q'),n)}}}
\qquad
\qquad
\forall j \in [0,2^n)\,.\,\aket{j}{n}\aket{0}{1}\to\aket{j}{n}\qket{\frac{(2j+1)r}{2^n}}
$
\end{center}
}

In doing so, we can adapt the same strategy in the modular exponentiation state preparation, by cutting off $\iseq{\inew{\overline{q}}}{\iseq{\inew{q'}}{\ihad{\overline{q}}}}$, so we should mainly focus on validating the portion shown above on the left.
The correctness specification is described above on the right.
In validating this portion, we notice that the qubit array $\overline{q}$ is in $\thadt$ type, and we generate random variables, with binary values $0$ or $1$, for every qubit in the array. Through the validation procedure in validating the modular exponentiation program above, we can assure that the result of the program produces the superposition state $\varphi_4$.

\subsubsection{Repeat-until-success Programs.}\label{sec:repeat-success}
We then examine the class of repeat-until-success programs, using the strategy in \Cref{sec:rand-testing}.
We show the Hamming weight and the distinct element state preparation program validation below.

\begin{figure}[t]
{\scriptsize
\begin{center}
$\begin{array}{c}
  \Qcircuit @C=0.5em @R=0.5em {
    &                     & & \gate{H} & \ctrl{6} & \qw & \qw & \qw & \qw & \qw & \qw & \qw & \\
    &                     & & \gate{H} & \qw & \ctrl{5} & \qw & \qw & \qw & \qw & \qw &  \qw & \\
    & \push{\overline{q_1}:\ket{0}\quad} & &  & & & &  & & & & & \push{\quad\varphi_5}\\
    &                     & & \dots & & & & \dots & & & & & \\
    &                     & & & & & & & & & & & \\
    &                     & & \gate{H} & \qw & \qw & \qw & \qw & \qw & \ctrl{1} & \qw & \qw & \\
    &                     & & \qw & \multigate{4}{\texttt{$\iadd{\overline{q_2}}{1}$}} & \multigate{4}{\texttt{$\iadd{\overline{q_2}}{1}$}} & \qw & \qw & \qw & \multigate{4}{\texttt{$\iadd{\overline{q_2}}{1}$}} & \qw & \meter & \\
    & \push{\overline{q_2}:\ket{0}\quad } & & \qw & \ghost{\texttt{$\iadd{\overline{q_2}}{1}$}} & \ghost{\texttt{$\iadd{\overline{q_2}}{1}$}} & \qw & \qw & \qw & \ghost{\texttt{$\iadd{\overline{q_2}}{1}$}} & \qw &  \meter & \push{\quad x}\\
    &                     & & \dots &                                     &                                     &     & \dots &     &                                         & & \dots & \\
    &                     & &       &                                     &                                     &     &       &     &                                         & & & \\
    &                     & & \qw   & \ghost{\texttt{$\iadd{\overline{q_2}}{1}$}} & \ghost{\texttt{$\iadd{\overline{q_2}}{1}$}} & \qw & \qw   & \qw & \ghost{\texttt{$\iadd{\overline{q_2}}{1}$}} & \qw & \meter 
    \gategroup{1}{3}{6}{3}{1em}{\{}
    \gategroup{7}{3}{11}{3}{1em}{\{}
    \gategroup{1}{12}{6}{12}{1em}{\}}
    }
    \end{array}
$
\end{center}
}
\caption{One-step Hamming weight state preparation (repeat-until-success).}
\label{fig:hammingcircuit}
\end{figure}

\noindent\textbf{\textit{Hamming Weights State Preparation.}}\label{sec:hammingweight}
Some algorithms \cite{10.5555/2011430.2011431} require a $k$-th Hamming weight superposition state preparation,
i.e., we prepare a state $\varphi_5=\frac{1}{\sqrt{N}}\sum_j^N \aket{c_j}{n}$, with the number of $1$'s bit in $c_j$ is $k$.
Assuming that $\varphi_5$ is a length $n$ qubit state, $\varphi_5$ has $N$ different basis-kets, with $N=\begin{pmatrix}
n\\
k
\end{pmatrix}$.

{\scriptsize
\begin{center}
$\hspace*{-0.5em}
\begin{array}{l@{\;}c@{\;}l}
Q(\overline{q_1},\overline{q_2})(j) &\triangleq& \ictrl{(\overline{q_1}[j])}{\cn{add}(\overline{q_2},1)}
\\[0.2em]
P(n,k) &\triangleq& \iseq{\inew{\overline{q_1}}}{\iseq{\iseq{\inew{\overline{q_2}}}{\ihad{\overline{q_1}}}}{\iseq{Re(Q(\overline{q_1},\overline{q_2}),n)}{\smea{x}{\overline{q_2}}{\sifb{x\cn{=} k}{\sskip}{P(n,k)}}}}}
\end{array}
$
\end{center}
}

The above is a repeat-until-success program of the Hamming weight program, with the circuit in \Cref{fig:hammingcircuit} showing a single quantum step in $P(n,k)$.
The program starts with two new length $n$ qubit arrays $\overline{q_1}$ and $\overline{q_2}$, and turns $\overline{q_1}$ to a uniform superposition by applying $n$ \cn{H} gates. We then repeat $n$ times of a controlled addition ($\ictrl{(\overline{q_1}[j])}{\cn{add}(\overline{q_2},1)}$) applications.
The controlled additions count the number of $1$'s bits in $\overline{q_1}$ and store the result in $\overline{q_2}$.
If the measurement on $\overline{q_2}$ results in $k$ (assigning to $x$), it means that the $\varphi_5$ state of the qubit array $\overline{q_1}$ is a superposition of basis-ket states with the vector having $k$ bits of $1$.
Otherwise, we repeat the process $P$ with two new qubit arrays $\overline{q_1}$ and $\overline{q_2}$ until the measurement result $k$ appears.
Note that $Q(\overline{q_1},\overline{q_2})$ in $Re$ is a function taking in a natural number argument $j$ and then performing a controlled addition.

{\small
\begin{center}
$
\begin{array}{c}
{\iseq{Re(Q(\overline{q_1},\overline{q_2}),n)}{\smea{x}{\overline{q_2}}{\sifb{x=k}{\sskip}{\sskip}}}}
\\[0.2em]
\forall j \in [0,2^n)\,.\,\aket{j}{n}\aket{0}{n}\to\aket{j}{n}\wedge \cn{sum}(\cn{n2b}(j))=k
\end{array}
$
\end{center}
}

To validate the correctness of the Hamming weight program, we shrink the program by removing the \cn{new} and \cn{H} operations.
The program piece and the transformed correctness specification are listed above.
We then utilize the procedure in \Cref{sec:rand-testing} to perform the validation.
Here, we assume that the $\thadt$ typed qubits $\overline{q_1}$ are already prepared, and we randomly generate a length $n$ bitstring for the random variables $\overline{q_1}[0],...,\overline{q_1}[n\sminus 1]$.
Each random variable, possibly being $0$ or $1$, represents the basis-bit of a single qubit superposition.
We set up the PBT to randomly sample values for the random variables and exclusively test the correctness of the transition behavior of basis-ket states.
The key correctness property ($\cn{sum}(\cn{n2b}(j))=k$) for the Hamming weight state is that each output basis-ket of $\overline{q_1}$ should have exactly $k$ bits of $1$.

The program efficiency can be easily assessed by counting the number of basis kets in a superposition quantum state.
Notice that every superposition state prepared by a simple Hadamard operation produces a uniform superposition, meaning that the likelihood of measuring out any basis-ket vector is equally likely.
Thus, we only need to compare the ratio between the number of basis-kets after the Hadamard operations are applied and the basis-ket number in $\overline{q_1}$ after the measurement is applied. The former contains $2^n$ different basis-kets for $n$ Hadamard operations, and the latter has $\begin{pmatrix}
n\\
k
\end{pmatrix}$ basis-kets in a $k$-th Hamming weight state.
So, the success rate of a single try in the program is $\begin{pmatrix}
n\\
k
\end{pmatrix} / 2^n$.

\begin{figure}[t]
{\hspace*{-0.2em}
\begin{minipage}[t]{0.3\textwidth}
\subcaption{Subroutine $Q(j,k)$ ; $C=\qbool{\overline{q_{j\splus 2}}}{=}{\overline{q_{k\splus 2}}}{q_0}$.}
\label{fig:subroutine}
\vspace*{0.3em}
{\tiny
$\begin{array}{l}
  \Qcircuit @C=0.5em @R=0.5em {
    & \push{\overline{q_{j\splus 2}}}  & &  \multigate{2}{\texttt{$C$}} & \qw &  \multigate{2}{\texttt{$C$}} & \qw & \\
    &  \push{\overline{q_{k\splus 2}}} & & \ghost{\texttt{$C$}}         & \qw & \ghost{\texttt{$C$}} & \qw &  \\
    & \push{q_0}             & & \ghost{\texttt{$C$}}         & \ctrl{1} & \ghost{\texttt{$C$}} & \qw &  \\
    & \push{\overline{q_1}}  & & \qw                                                                       & \gate{\texttt{$\iadd{\overline{q_1}}{1}$}} & \qw & \qw &  \\
    }
    \end{array}
$
}
\end{minipage}
  \begin{minipage}[t]{0.64\textwidth}
   \subcaption{Overall One Step Circuit}
\label{fig:diselems}
\vspace*{0.3em}
{\tiny
$\begin{array}{c}
  \Qcircuit @C=0.5em @R=0.5em {
    &\push{\overline{q_2}:\ket{0}}                     & & \gate{\cn{H}} &  \multigate{7}{\texttt{$Q(2,3)$}} & \qw & \qw & \qw & \multigate{7}{\texttt{$Q(2,n\splus 1)$}} & \multigate{7}{\texttt{$Q(3,4)$}}  & \qw & \qw & \qw     & \multigate{7}{\texttt{$Q(n,n\splus 1)$}} & \qw & \\
    &\push{\overline{q_3}:\ket{0}} & & \gate{\cn{H}} & \ghost{\texttt{$Q(2,3)$}} & \qw & \qw & \qw & \ghost{\texttt{$Q(2,n\splus 1)$}} & \ghost{\texttt{$Q(3,4)$}}   & \qw & \qw & \qw     & \ghost{\texttt{$Q(n,n\splus 1)$}} &  \qw & \\
    &  & &  & & &   & & & & & & & & & \push{\varphi_6}\\
    &                     & & \dots & & & \dots & & & &  & \dots & & & &\\
    &                     & & & & & & & & & & & & &\\
    &\push{\overline{q_{n\splus 1}}:\ket{0}}                     & & \gate{\cn{H}} & \ghost{\texttt{$Q(2,3)$}} & \qw & \qw & \qw & \ghost{\texttt{$Q(2,n\splus 1)$}} & \ghost{\texttt{$Q(3,4)$}}  & \qw & \qw & \qw     & \ghost{\texttt{$Q(n,n\splus 1)$}} & \qw & \\
    &\push{q_0:\ket{0}}    & & \qw & \ghost{\texttt{$Q(2,3)$}} & \qw & \qw & \qw & \ghost{\texttt{$Q(2,n\splus 1)$}} & \ghost{\texttt{$Q(3,4)$}}  & \qw & \qw & \qw     & \ghost{\texttt{$Q(n,n\splus 1)$}} & \qw & \\
    &\push{\overline{q_1}:\ket{0}}  & & \qw & \ghost{\texttt{$Q(2,3)$}} & \qw & \qw & \qw & \ghost{\texttt{$Q(2,n\splus 1)$}} & \ghost{\texttt{$Q(3,4)$}}  & \qw & \qw & \qw     & \ghost{\texttt{$Q(n,n\splus 1)$}} & \meter & \push{x}   
    %\gategroup{1}{3}{6}{3}{1em}{\{}
    \gategroup{1}{15}{7}{15}{1em}{\}}
    }
    \end{array}
$
}
\end{minipage}
\caption{One step distinct element state preparation.}
\label{fig:distinctelem}
}
\end{figure}

\myparagraph{Distinct Element State Preparation.}\label{sec:distinctness}
Another special superposition state is the one in the element distinctness algorithm.
Here, we assume that we are given a graph with $n$ different vertices, and the algorithm begins with a superposition of different combinations of vertices, as shown below.

{\small
\begin{center}
$
\varphi_6=\frac{1}{\sqrt{n!}}\sum_{j} \sigma_j(\ket{x_1}\ket{x_2}...\ket{x_n})
$
\end{center}
}

Here, $x_1$, $x_2$, ..., $x_n$ are different vertex keys in the graph, $\sigma_j$ is a permutation of the key list $\ket{x_1}\ket{x_2}...\ket{x_n}$.
There are $n!$ different kinds of permutations, so the uniform amplitude for each basis-ket is $\frac{1}{\sqrt{n!}}$.
Such a permutation state means we are preparing a superposition of all permutations of the different vertex keys.
Such a superposition state is widely used in many algorithms, such as the quantum fingerprinting algorithm \cite{Buhrman_2001}.

{\scriptsize
\begin{center}
$\hspace*{-0.5em}
\begin{array}{l@{\;}c@{\;}l}
Q(k)(j) &\triangleq& \iseq{\qbool{\overline{q_{j+2}}}{=}{\overline{q_{k+2}}}{q_0}}{\iseq{\ictrl{q_0}{(\overline{q_1}+1)}}{\qbool{\overline{q_{j+2}}}{=}{\overline{q_{k+2}}}{q_0}}}
\\[0.2em]
R(j)(n) &\triangleq& Re(Q(j),n)
\\[0.2em]
H(j) &\triangleq& \ihad{\overline{q_{j+2}}}
\\[0.2em]
T(j) &\triangleq& \inew{\overline{q_{j+1}}}
\\[0.2em]
P(n) &\triangleq& \iseq{\inew{q_0}}{\iseq{Re(T,n\splus 1)}{\iseq{{Re(H,n)}}{\iseq{Re(R(n\sminus 1),n)}{\smea{x}{\overline{q_1}}{\sifb{x\cn{=} 0}{\sskip}{P(n)}}}}}}
\end{array}
$
\end{center}
}

For simplicity, we only implement the above program to prepare a superposition state of distinct elements, i.e., each basis-ket in the superposition state stores $n$ distinct elements (vertex key), each key having a qubit size $m$. Note that if $n=2^m$, i.e., we have $2^m$ different vertices having keys $u\in[0,2^m)$, then the superposition state represents a superposition of all the permutations.
We show a repeat-until-success program for preparing such a state above, with the circuit in \Cref{fig:distinctelem} showing a single quantum step in $P(n)$. We first initialize a single qubit $q_0$, and use $T(j)$ to initialize $n+1$ different qubit arrays, $\overline{q_{j+1}}$, with $j\in [0,n+1)$, and we assume that $\overline{q_{j+1}}$ is an $m$ length qubit array.
We then apply \cn{H} gates to all qubit arrays $\overline{q_{j+2}}$ ($j \in [0,n)$). $q_0$ and $\overline{q_1}$ are ancillary qubits.

Essentially, we can view $\overline{q_{j+2}}$ ($j \in [0,n)$) as an $n$-length array of qubit arrays.
The program applies $O(n^2)$ times of $Q$ processes, each of which applies an equivalent check on two elements in the $n$-length array,
i.e., we compare the basis-ket data in $q_{j+2}$ and $q_{k+2}$ ($j,k\in[0,n)$) to see if their content are equal, and store the boolean result in $q_0$ bit.
Then, we also add the result to $\overline{q_1}$ and apply the comparison circuit again to clean up the ancillary qubit $q_0$, meaning that we restore $q_0$'s state to $\ket{0}$. This procedure describes the circuit in \Cref{fig:subroutine}.

After we apply the $Q$ function to any two different elements in the $n$-length array, we observe that $q_0$ is back to the $\ket{0}$ state, and $\overline{q_1}$ stores the number of same basis-kets between any two distinct elements in the $n$-length array.
We then measure $\overline{q_1}$ and see if the result is $0$. If so, a permutation superposition state is prepared because it means that in all the basis-kets in the prepared superposition, there are no two-qubit array elements $\overline{q_k}$ and $\overline{q_l}$ that have equal key.
If not, we repeat the process, and the repeat-until-success program guarantees the creation of the permutation superposition state.

We use the procedure in \Cref{sec:rand-testing} to validate the correctness. We create the validating program by removing the \cn{new} and \cn{H} operations from the original program. The program piece and the transformed specification are listed below.

{\small
\begin{center}
$
\begin{array}{c}
\iseq{Re(R(n-1),n)}{\smea{x}{\overline{q_1}}{\sifb{x=0}{\sskip}{\sskip}}}
\\[0.2em]
x=0\Rightarrow\forall j,j'\in [0,2^{n * m})\,.\,\aket{0}{1}\aket{0}{n}\aket{j}{n*m}\to \aket{0}{1}\aket{j'}{n*m} \wedge \cn{dis}(j',m)
\end{array}
$
\end{center}
}

Here, $j$ and $j'$ represent the values for two qubit arrays, i.e., $j$ represents the bitstring value for composing basis-ket values of all elements in the qubit array $\overline{q_{l\splus 2}}$ ($l\in[0,n)$). The qubit array $\overline{q_{j\splus 2}}$ has $n \cn{*} m$ qubits, and we slice the basis-vector for the whole qubit array into $n$ different small segments for the qubit ranges $[l\cn{*} m,l\cn{*} (m\splus 1)\sminus 1)$, each basis-vector segment representing a vertex key.
Since we apply Hadamard operations to all of them, it creates a uniform superposition state containing $2^{n \cn{*} m}$ different basis-vector states.
In the post-state of the specification, the $q_0$ qubit is still $\ket{0}$.
For the qubit arrays $\overline{q_{2}},...,\overline{q_{m\splus 2}}$, if the measurement result is $x=0$, we result in a superposition state of distinct elements, i.e., any two elements (each element $l$ is a segment of $[l\cn{*} m,l\cn{*} (m\splus 1)\sminus 1)$) in the qubit array $\overline{q_{j\splus 2}}$ have distinct basis-vectors.
We use the predicate $\cn{dis}(j',m)$ to indicate that all length $m$ segments in the bitstring $\ket{j}$ are pairwise distinct.
Via our PBT framework, we have high assurance that the distinct element state preparation program correctly prepares the superposition state.

The program's efficiency in preparing the superposition state can be easily assessed by counting the number of basis-ket states in the superposition.
Notice that every superposition state prepared by a simple Hadamard operation produces a uniform superposition, meaning that the likelihood of measuring out any basis state vectors is equally likely.
Thus, we only need to compare the ratio between the number of basis-kets right after the Hadamard operations are applied and the basis-ket number in $\overline{q_{j+2}}$.
Here, we count the case for $n=2^m$ where a permutation superposition state is prepared.
For $j\in[0,n)$ with $n=2^m$, after the measurement is applied. The former contains $2^{n * m}$ different basis-kets for $n * m$ Hadamard operations, and the latter has $n!$ basis-kets in a $k$-th permutated superposition state. So, the success rate of a single try in the program is $\frac{n !}{2^{n * m}}$.

\section{Discussion: Efficiency, Scalability, and Utility}\label{sec:eval}

We compare QSV's efficiency and scalability with the state-of-the-art platforms. We show that QSV can effectively validate program properties, providing confidence in program correctness. We also show that QSV \textit{scales} to validate and realize state preparation programs regardless of the size (in terms of qubit number required).
We thus justify the QSV's utility in successfully capturing bugs.

\begin{figure}[t]
{\scriptsize
\begin{center}
\begin{tabular}{| l | c | c | c | c |  c |}
\hline
 Program  & QSV  QCT 8B & QSV QCT 60B  & Qiskit Sim 8B & Qiskit Sim 60B & DDSim Sim 60B \\
 \hline
$n$ basis-ket & $<$ 1.5  & 2.4 & $<$ 1.5 & No & No \\
Modular Exp. & $<$ 1.5  & 19.7 & No & No & No\\
Amplitude Amp. & $<$ 1.5  & 2.1 & $<$ 1.5 & No & No  \\
Hamming Weight & $<$ 1.5  & 24.9 &  $<$ 1.5 & No & No \\
Distinct Element & $<$ 5  & 336  & No & No & No \\
\hline                           
\end{tabular}
\end{center}
}
\caption{Evaluation on different state preparation programs for 8/60 qubit single registers (8B/60B). "QCT" : time (in seconds) for QuickChick to run 10,000 tests. "Sim": time (in seconds) or if Qiskit/DDSim can execute a single test. }
\label{fig:self-data}
\end{figure}

\myparagraph{Efficiency.}\label{sec:testefficient}
We discuss the program development procedure in QSV, compared to other systems such as Qiskit.
Discussing the efficiency of a validation framework needs to be put in the context of human efforts for program development,
as users mainly care about how to effectively use QSV to develop programs.

The general procedure for developing state preparation programs in QSV is the traditional test-driven program development.
We first present the program correctness properties in the superposition state format for different programs, such as the one in \Cref{sec:intro,sec:evaluation}.
We then start implementing the program using a possible program pattern and see if we can write the correct program based on it, where the \pqasm high-level abstraction helps us write programs.
For example, in dealing with all the programs in \Cref{fig:qiskit-data}, we first try to see if we can write all these programs via the quantum loop program pattern,
and rewrite the correctness properties based on the strategy presented in \Cref{sec:quantumloop}.

We then run the QSV validator to validate the implemented programs against the properties.
After several rounds of corrections, one can typically judge whether the program is implementable.
For example, using our validator, we found that implementing the Hamming weight and distinct element programs based on the quantum loop program pattern might be hard.
We then switch to other program patterns to implement these programs; e.g., we use the repeat-until-success program pattern to implement the two programs by rewriting their correctness properties to those in \Cref{sec:repeat-success} and successfully find a solution.
Our validator can effectively validate a program using our PBT framework, which generates $10,000$ test cases each time.
As shown in \Cref{fig:self-data}, running $10,000$ randomly generated test cases for all our example programs takes less than $5$ seconds for small-size programs ($8$ qubits) and less than 5.5 minutes for large-size programs ($60$ qubits).
This indicates that a program developer can quickly correct minor bugs when developing their programs.

On the other hand, developing state preparation programs in the state-of-the-art system might be painful, e.g., it is unlikely one can perform test-driven development to implement the programs in \Cref{fig:qiskit-data}, mainly because of a lack of proper validation facilities.
As we can see in \Cref{fig:self-data}, Qiskit might not execute a single test for some small-size ($8$ qubit) programs, such as distinct-element programs.
The modular exponentiation program is executed for some small number settings, but is not executable in general because the Qiskit simulator applies special optimizations to components of Shor's algorithm.
Executing a large program ($60$ qubits) is completely impossible; either the program is too large for IBM's Aer simulator to run at all, and the simulator errors out, or the circuit construction may take hours and still not be done.
This indicates the difficulty of developing large-scale programs by conducting small-scale testing in Qiskit, not to mention the need to validate large datasets and coverage.
Another key issue is that the high-level abstraction support in the state-of-the-art systems is not well provided.
In Qiskit, we can only find quantum addition operations; the other arithmetic and comparison operations are missing.
In fact, we implement these operations in Qiskit based on our implementation in QSV.

\myparagraph{Scalability.}
To evaluate the QSV's scalability, we not only compare its execution at small and large sizes with Qiskit but also with another state-of-the-art quantum simulator, DDSim.
Here, we attempted to recreate (or find existing implementations of) the aforementioned programs on DDSim and Qiskit. We also performed PBT on our \pqasm implementations on systems of various sizes and verified the rigidity of our tests by mutating either the properties or the states and verifying that the tests failed. 

As in \Cref{fig:self-data}, we have fully validated the five examples in these papers via our PBT framework.
As far as we know, these constitute the first validated-correct implementations of the $n$ basis-ket, Hamming weight, and distinct elements programs. 
All other operations in the figure were validated with Quick\-Chick. To ensure these tests were effective, besides our program development procedures above, we also confirmed they could find hand-injected bugs; e.g., we changed the rotation angles in the \cn{Ry} gate in the amplitude amplification state preparation and performed some mutation testing (e.g. replacing some non-skip terms with $\sskip$ (skip) operations) and confirmed that our PBT could catch the inserted bugs.
The tables in \Cref{fig:self-data} give the running times for our validator to validate programs---the times include the cost of compiling the Rocq code and running it with $10,000$ randomly generated inputs via QuickChick.
We validated these programs on small (8-qubit) and large (60-qubit) inputs (the numbers relevant to the reported qubit and gate sizes in \Cref{fig:qiskit-data}), with all validation completed within 2.5 minutes (most within seconds).
%Most tests are completed in a few seconds, while the test for the distinct element superposition preparation program finishes in a few minutes. 
For comparison, we translated our programs to \sqir, converted the \sqir programs to OpenQASM 2.0 \cite{Cross2017}, and then attempted to simulate the resulting circuits on a \textit{single test input} using DDSim~\cite{ddsim}, a state-of-the-art quantum simulator, and listed the result in the fifth column. Unsurprisingly, the simulation of the 60-bit versions did not complete when running overnight.
The third and fourth columns in \Cref{fig:self-data} show the results for executing a \textit{single program run} in Qiskit, and Qiskit executes a few small-size programs (\Cref{sec:testefficient}) and none of the large-size programs.
The experiment provided strong assurance of QSV's scalability.

\begin{wrapfigure}{r}{5.2cm}
{\scriptsize
\begin{center}
\begin{tabular}{| l | c |}
\hline
 Operation  & QCT  \\
 \hline
Addition & 2  \\
Comparison & 5  \\
Modular Multiplication & 794 \\
\hline                           
\end{tabular}
\end{center}
}
\caption{Arith OP QC time (60B).}
%\vspace*{-1em}
\label{fig:runningtime}
\end{wrapfigure}

There is a difference in the program execution between DDSim and QSV.
The latter abstracts arithmetic operations and assumes they can be handled by the previous VQO~\cite{oracleoopsla} framework, whereas DDSim executes the entire circuits generated from a QSV program.
To compare the effects, we list the QuickChick testing time (running 10,000 tests) for the operations used in our state preparation programs in \Cref{fig:runningtime}; these running time data were provided in VQO.
The addition and comparison circuits do not greatly affect the execution of our \pqasm programs.
Validating modular multiplication circuits might be costly, as our $60$-bit modular exponentiation involves $60$ modular multiplication operations.
However, a typical validation scheme might only validate the correctness of a costly subcomponent once and then use its semantic properties to validate other programs that use it.
More importantly, the purpose of the experiment of DDSim/Qiskit executions is to show the state-of-the-art impossibility of executing quantum programs on a classical computer, while our QSV framework can validate quantum programs.

\myparagraph{Utility.}
One of the utilities of a program validation framework is to find bugs or faults in the existing algorithms.
During the development of the above programs, we identified several issues in two original algorithms \cite{Buhrman_2001,1366221} that utilize these special superposition states. Both algorithms require preparing a superposition state of distinct elements (or of permutations of distinct elements), but they do not specify how to effectively prepare such a state. To the best of our knowledge, the state preparation program in \Cref{sec:distinctness} is the first program implementation of the state via the repeat-until-success scheme. As our probability analysis shows, the likelihood of preparing such a state is not very high.
This fact might indicate that the quantum algorithms' advantage arguments over classical algorithms in these works are not solid, given the unclear preparation of the initial states. Without our implementations of these state preparation programs, it is impossible to detect these subtle potential faults in these algorithms.

Indeed, in the algorithm \cite{10.5555/2011430.2011431} that uses the initial Hamming-weight superposition state, the authors recognized the low probability of preparing the initial state via the repeat-until-success scheme and proposed using a specialized gate instead of Hadamard gates to start their repeat-until-success state preparation program. The special gates created a simple superposition state with a different probability distribution, unlike the uniform distribution produced by Hadamard gates. The analysis of these specialized superposition gates with different probability distributions will be included in our future work.

Another utility of using QSV is to judge the implemented programs' correctness and find a more optimized implementation. Other frameworks, such as OpenQASM, do not have property-based testing features.

\section{Related Work}
\label{sec:related}

This section gives related work beyond the discussion in \Cref{sec:implementation}.

\noindent\textbf{\textit{Quantum Circuit Languages.}}
Prior research has developed circuit-level compilers to compile quantum circuit languages to quantum computers, such as Qiskit \cite{Qiskit2019}, \tket \cite{tket}, Staq \cite{Amy2020}, PyZX \cite{Kissinger2019}, Nam \emph{et al.} \cite{Nam2018}, quilc~\cite{quilc}, Cirq~\cite{cirq}, ScaffCC \cite{JavadiAbhari2015}, and Project Q~\cite{Steiger_2018}. 
In addition, many quantum programming languages have been developed in recent years. 
Many of these languages (e.g. Quil~\cite{quilc}, OpenQASM ~\cite{Cross2017,10.1145/3505636}, \sqir~\cite{VOQC}) describe low-level circuit programs.
Higher-level languages may provide library functions for performing common oracle operations (e.g., Q\# \cite{qsharp}, Scaffold~\cite{scaffold,scaffCCnew}) or support compiling from classical programs to quantum circuits (e.g., Quipper~\cite{10.1145/2491956.2462177}), but still leave some important details 
(like deallocating extra intermediate qubits) to the programmer.
There has been some work on type systems to enforce that deallocation happens correctly (e.g., Silq~\cite{sliqlanguage}) and on automated insertion of deallocation circuits (e.g., Quipper~\cite{10.1145/2491956.2462177}, Unqomp~\cite{unqomp}), but while these approaches provide useful automation, they may also lead to inefficiencies in compiled circuits.

\noindent\textbf{\textit{Quantum Software Testing and Validation.}}
There have been many approaches developed for validating quantum programs \cite{morphq_bugs,fuzz4all,10.1109/ASE51524.2021.9678798,fortunato,long:24,QDiff} including the use differential ~\cite{QDiff} and metamorphic testing ~\cite{10.1109/ICSE48619.2023.00202}, and mutation testing ~\cite{fortunato} and fuzzing ~\cite{fuzz4all}. Some key challenges exist for testing quantum programs. First, their input space explodes due to superposition. Second, their results are probabilistic (meaning we need to use statistical measures and/or other approaches to evaluate results).  Last, the expected result may be difficult or even impossible to determine. To date, testing approaches have focused on validating small circuit subroutines (i.e., with limited input qubit size) rather than comprehensive quantum programs, and they are all limited to testing in Qiskit, which might not capture all machine limitations.

\noindent\textbf{\textit{Methodologies Possibly Used for Validating Quantum Programs.}}
SymQV \cite{10.1007/978-3-031-27481-7_12} proposed a method of encoding quantum states and gates as SMT-solvable predicates to perform automated verification.
Chen \emph{et al.} \cite{10.1145/3591270} and Abdulla \emph{et al.} \cite{abdulla2024verifyingquantumcircuitslevelsynchronized} used tree automata to symbolize quantum gates, instead of quantum states, and utilized tree automata to construct a tree structure for easing automated verification. These works can handle some large programs, but these programs have simple program structures, such as QFT.
Mei \emph{et al.} \cite{10.1007/978-3-031-65633-0_25} performed quantum stabilizer simulation based on the Gottesman–Knill theorem, which is a small subset of quantum programs and mainly used for error correction programs. Quasimodo \cite{10.1007/978-3-031-37709-9_11} is another symbolic execution based on a BDD-like structure to symbolize gates rather than states; their results are similar to the tree automata works \cite{abdulla2024verifyingquantumcircuitslevelsynchronized,10.1145/3591270}.
\qafny \cite{li2024,10.1145/3763157} transformed quantum program verification to Dafny for automated verification.
These works tried to transform states and gates to perform automated verification, different from QSV, which tries to perform program testing and validation. The methodologies are also different from QSV where they try to symbolize quantum gates and states, while QSV only inserts special treatments in the standard quantum state representations. As a result, QSV can deal with large programs with comprehensive program structures.

\noindent\textbf{\textit{Verified Quantum Compilers.}}
Recent work has looked at verified optimization of quantum circuits (e.g., \voqc~\cite{VOQC}, CertiQ~\cite{Shi2019}), but the problem of verified \emph{compilation} from high-level languages to quantum circuits has received less attention.
The only examples we know of verified compilers for quantum circuits are ReVerC~\cite{reverC}, ReQWIRE~\cite{Rand2018ReQWIRERA}, and Feynman \cite{Amy_Feynman}
\cite{amyphdthesis}. 
ReVerC and ReQWIRE support verified translation from a low-level Boolean expression language to circuits consisting of \texttt{X}, \texttt{CNOT}, and \texttt{CCNOT} gates.  Feynman verifies the correctness of compiled circuits relative to an input circuit or a path-integral specification of a circuit, while QSV validates whether programs have specified properties before they are compiled to the circuit level. It’s feasible that a programmer can use QSV to validate the program and then use Feynman to verify the compiled circuit relative to a circuit or path integral specification of the program. 
VQO \cite{oracleoopsla} is a certified compilation framework for verifying and compiling quantum arithmetic operations.
QSV utilizes VOQC \cite{VOQC} and VQO \cite{oracleoopsla} to compile \pqasm programs to quantum circuits.

\section{Conclusion and Limitations}
\label{sec:conclusion}

We present QSV, a framework for expressing and automatically validating quantum state preparation programs.
The core of QSV is a language \pqasm, which can express a restricted class of quantum programs that are efficiently testable for certain properties and are useful for implementing state preparation programs. 
We have verified the translator from \pqasm to \sqir and have validated (or randomly tested) many programs written in \pqasm.
We have used \pqasm to implement state preparation programs useful in quantum computation, such as the ones in \Cref{fig:qiskit-data}.
We hope this work will be the basis for building a quantum validation framework for validating quantum programs on classical computers.

QSV is capable of validating properties of many quantum programs. As mentioned in \Cref{sec:intro}, QSV targets validating state preparation programs. For many quantum programs, QSV is able to validate the most significant part of the program. For example, the validated modular multiplication program in \Cref{fig:mod-mult} is essentially 90\% of Shor's algorithm, except for the final inverse QFT gate and measurement.

Our type system specifically locates Hadamard operations, but there are no actual restrictions on Hadamard operations, as users can easily define similar behaviors via our \cn{Ry} and oracle operations. Via our type system, we identify the beginning Hadamard operations as a general quantum algorithm component to generate superposition sources so that QSV can locate the places to create random inputs for validating programs. We recognize the superposition state generation in many quantum algorithms as the major bottleneck for testing quantum programs; therefore, we utilize types to identify them, with special treatment to transform the superposition states to a simple and testable format.

\section{Data-Availability Statement}
The Rocq proofs and the experiment in the paper are available on Zenodo: \url{https://doi.org/10.5281/zenodo.15073729} \cite{li_2026_18568946}. The Zenodo artifact can be run in the experimental setting described in its README. There is also a GitHub repository at \url{https://github.com/qafny/pqasm}.
%% Bibliography
\bibliography{reference}
\newpage
 \appendix
 \section{Additional PQASM Semantics and Typing Rules}\label{appx:pqasm}

\begin{figure*}[h]
{\scriptsize
  \begin{mathpar}
        \inferrule[S-SeqC]{(\Phi,e_1) \xrightarrow{r} (\Phi', e'_1) }
        {(\Phi,\sseq{e_1}{e_2}) \xrightarrow{r} (\Phi',\sseq{e'_1}{e_2}) }

        \inferrule[S-SeqT]{}
        {(\Phi,\sseq{\sskip}{e_2}) \xrightarrow{1} (\Phi,e_2) }

      \inferrule[S-IfT]{}{ (\Phi,\sifb{\cn{true}}{e_1}{e_2}) \xrightarrow{1} (\Phi,e_1)}
 
       \inferrule[S-IfF]{}{ (\Phi,\sifb{\cn{false}}{e_1}{e_2}) \xrightarrow{1} (\Phi,e_2)}
  \end{mathpar}
}
\caption{Remaining program Level \pqasm rules.}
\label{fig:exp-semanticsa}
\end{figure*}

\begin{figure}[h]
{\footnotesize
  \begin{mathpar}
     \inferrule[Seq]{\Omega\vdash_g \instr_1\triangleright \Omega' \\ \Omega'\vdash_g \instr_2\triangleright \Omega''}{\Omega \vdash_g \iseq{\instr_1}{\instr_2}\triangleright \Omega''} 

       \inferrule[ESeq]{\Sigma;\Omega\vdash e_1\triangleright \Omega' \\ \Sigma;\Omega'\vdash e_2\triangleright \Omega''}{\Sigma;\Omega \vdash \iseq{e_1}{e_2}\triangleright \Omega''} 
      
      \inferrule[TIf]{\Sigma \vdash B \\ \Sigma;\Omega\vdash e_1 \triangleright \Omega'\\\Sigma;\Omega\vdash e_2 \triangleright \Omega'}{\Sigma;\Omega\vdash \sifb{B}{e_1}{e_2} \triangleright \Omega'} 
  \end{mathpar}
}
  \caption{Remaining typing rules.}
  \label{fig:exp-well-typeda}
\end{figure}

\Cref{fig:exp-semanticsa} shows additional program-level semantic rules.
Rules \rulelab{S-IfT} and \rulelab{S-IfF} perform classical conditional evaluation.
In \pqasm, the classical variables are evaluated via a substitution-based approach, as in \rulelab{S-Mea}.

\Cref{fig:exp-well-typeda} shows additional typing rules.
In our type system, rule \rulelab{TIf} ensures that any variables mentioned in $B$ are properly scoped by the constraint $FV(B)\subseteq \Sigma$ \footnote{The rule is parameterized by different $B$ formalism.}, i.e., all free variables in $B$ are bounded by $\Sigma$.
Note that the two branches of \rulelab{TIf} have the same output $\Sigma'$, i.e., the qubit manipulations in the two branches need to have the same effect.
If one branch has a measurement on a qubit, the other branch must have the same qubit measurement.

\section{Additional Rules for Translation from \pqasm to \sqir}\label{sec:vqir-compilationa}

\begin{figure}[h]
{\small
  \begin{mathpar}
    \inferrule[CRy]{}{\Xi \vdash \iry{r}{q} \gg (\gamma,\textcolor{blue}{\iry{r}{(\Xi(q))}})}
    
    \inferrule[CCU]{\Xi\vdash \instr \gg \textcolor{blue}{\epsilon}\\
      \textcolor{blue}{\epsilon' = \texttt{ctrl}(\gamma(q),\epsilon)}}{\Xi\vdash\ictrl{p}{\instr} \gg \textcolor{blue}{\epsilon'}}    
   
       \inferrule[CNext]{ \Xi \vdash \iota \gg \textcolor{blue}{\epsilon}}{(n,\Xi,\iota) \gg (n,\Xi,\textcolor{blue}{\epsilon})}             
    
     \inferrule[CHad]{}{(n,\Xi, \ihad{q}) \gg (n,\Xi, \ihad{\Xi(q)})}             
               
    \inferrule[CSeq]{ (n,\Xi, e_1) \gg (n',\Xi',\textcolor{blue}{\chi_1}) \\ (n',\Xi',e_2) \gg (n'',\Xi'',\textcolor{blue}{\chi_2})}{(n,\Xi,\sseq{e_1}{e_2}) \to (n'',\Xi'',\textcolor{blue}{\sseq{\chi_1}{\chi_2}})}             
  
    \inferrule[CNew]{\Xi'=\Xi[\forall q\in\overline{q}\,.\,q \mapsto \slen{\Xi}+\cn{ind}(\overline{q},q)]}{(n,\Xi, \inew{\overline{q}}) \gg (n+\slen{\overline{q}},\Xi', \sskip)}             
  
      \inferrule[CMea]{(n,\Xi \backslash \overline{q},e) \gg (n',\Xi',\textcolor{blue}{\chi})}{(n,\Xi,\smea{x}{\overline{q}}{e}) \to (n',\Xi',\textcolor{blue}{{\mathpzc{M}(\overline{q})};\chi})}

  \end{mathpar}
}
\caption{Select \pqasm to \sqir translation rules (\sqir circuits are marked blue). $\slen{\Xi}$: the length of $\Xi$; $\cn{ind}(\overline{q},q)$: the index of $q$ in array $\overline{q}$. $\textcolor{blue}{{\mathpzc{M}(\overline{q})}}$ repeats $\slen{\overline{q}}$ times on measuring qubit array $\overline{q}$.}
\label{fig:compile-vqir}
\end{figure}

\Cref{fig:compile-vqir} shows additional rules, which are the program-level rules, that translate a \pqasm program to a hybrid Rocq program including \sqir circuits with possible measurement operations.
Rule \rulelab{CNext} connects the instruction and the program-level translations.
Rule \rulelab{CHad} translates a Hadamard operation to a \sqir Hadamard gate, while rule \rulelab{CSeq} translates a sequencing operator.

Rule \rulelab{CNew} is one of the program-level rules, translating a qubit creation operation in \pqasm.
In \sqir, there is no qubit creation, in the sense that every qubit is assumed to exist in the first place.
The translation essentially translates the operation to a SKIP operation in \sqir and increments the qubit heap size in the generated \sqir program.
Note that the qubit size in the translation is always incrementing, i.e., a quantum measurement does not remove qubits but just makes some qubits inaccessible.
Rule \rulelab{CMea} translates a \pqasm measurement operation to \sqir measurements by repeatedly measuring out qubits in $\overline{q}$.
The translation removes measured qubits from $\Xi$, but it does not modify the qubit size.

\section{OQASM: An Assembly Language for Quantum Oracles}
\label{sec:vqir}

The \oqasm expression $\mu$ used in \Cref{fig:pqasm} places an additional wrapper on top of the \oqasm expression $\iota$ given in \Cref{fig:vqir}. Here, we first provide a step-by-step explanation of \oqasm.

\oqasm is designed to express efficient quantum
oracles that can be easily tested and, if desired, proved
correct.
\oqasm operations leverage both the standard
computational basis and an alternative basis connected by the quantum
Fourier transform (QFT). 
\oqasm's type system tracks the bases of variables in
\oqasm programs, forbidding operations that would introduce
entanglement. \oqasm states are therefore efficiently
represented, so programs can be effectively tested and are simpler to
verify and analyze. In addition, \oqasm uses \emph{virtual qubits}
to support \emph{position shifting operations}, which support
arithmetic operations without introducing extra gates during
translation. All of these features are novel to quantum assembly
languages. 

This section presents \oqasm states and the language's syntax,
semantics, typing, and soundness results. As a running example, the QFT
adder~\cite{qft-adder} is shown in \Cref{fig:circuit-exampleb}. The Rocq
function \coqe{rz_adder} generates an \oqasm program that adds two
natural numbers \coqe{a} and \coqe{b}, each of length \coqe{n} qubits.

\begin{figure*}[h]
  \centering
  \begin{tabular}{c @{\qquad} c}
  \begin{minipage}[b]{.45\textwidth}
  {\scriptsize
    \Qcircuit @C=0.5em @R=0.75em {
      \lstick{\ket{a_{n-1}}} & \qw & \ctrl{5} & \qw & \qw & \qw & \qw & \qw & \qw & \qw & \rstick{\ket{a_{n-1}}} \\
      \lstick{\ket{a_{n-2}}} & \qw & \qw & \ctrl{4} & \qw & \qw & \qw & \qw & \qw & \qw & \rstick{\ket{a_{n-2}}}\\
      \lstick{\vdots} & & & & & & & & & & \rstick{\vdots} \\
      \lstick{} & & & & & & & & & & \\
      \lstick{\ket{a_0}} & \qw & \qw & \qw & \qw & \qw & \qw & \ctrl{1} & \qw & \qw & \rstick{\ket{a_0}} \\
      \lstick{\ket{b_{n-1}}} & \multigate{5}{\texttt{QFT}} & \gate{\texttt{SR 0}} & \multigate{3}{\texttt{SR 1}} & \qw & \qw & \qw & \multigate{5}{\texttt{SR (n-1)}} & \multigate{5}{\texttt{QFT}^{-1}} & \qw & \rstick{\ket{a_{n-1} + b_{n-1}}} \\
      \lstick{} & & & & & \dots & & & & \\
      \lstick{\ket{b_{n-2}}} & \ghost{\texttt{QFT}} & \qw  &  \ghost{\texttt{SR 1}} & \qw & \qw & \qw & \ghost{\texttt{SR (n-1)}} & \ghost{\texttt{QFT}^{-1}} & \qw & \rstick{\ket{a_{n-2} + b_{n-2}}} \\
      \lstick{\vdots} & & & & & & & & & & \rstick{\vdots} \\
      \lstick{} & & & & & & & & & & \\
      \lstick{\ket{b_0}} & \ghost{\texttt{QFT}} & \qw & \qw & \qw & \qw & \qw & \ghost{\texttt{SR (n-1)}} & \ghost{\texttt{QFT}^{-1}}  & \qw & \rstick{\ket{a_0 + b_0}} 
      }
      }
  \subcaption{Quantum circuit}
  \end{minipage}
  &
  \begin{minipage}[b]{.4\textwidth}
  \begin{coq}
  Fixpoint rz_adder' (a b:var) (n:nat) 
    := match n with 
       | 0 => ID (a,0)
       | S m => CU (a,m) (SR m b); 
                rz_adder' a b m
       end.
  Definition rz_adder (a b:var) (n:nat) 
    := Rev a ; Rev b ; $\texttt{QFT}$ b ;
       rz_adder' a b n;
       $\texttt{QFT}^{-1}$ b; Rev b ; Rev a.
  \end{coq}
  \subcaption{\oqasm metaprogram (in Rocq)}
  \end{minipage}
  \end{tabular}
  \caption{Example \oqasm program: QFT-based adder}
  \label{fig:circuit-exampleb}
  \end{figure*}

\subsection{OQASM States} \label{sec:oqasm-states}

\begin{figure}[h]
  \small
  \[\hspace*{-0.5em}
\begin{array}{l c c l}
      \text{Bit} & b & ::= & 0 \mid 1 \\
      \text{Natural number} & n & \in & \mathbb{N} \\
      \text{Real} & r & \in & \mathbb{R}\\
      \text{Phase} & \alpha(r) & ::= & e^{2\pi i r} \\
      \text{Basis} & \tau & ::= & \texttt{Nor} \mid \texttt{Phi}\;n \\
      \text{Unphased qubit} & \overline{q} & ::= & \ket{b} ~~\mid~~ \qket{r} \\
      \text{Qubit} & q & ::= &\alpha(r) \overline{q}\\
      \text{State (length $d$)} & \varphi & ::= & q_1 \otimes q_2 \otimes \cdots \otimes q_d
    \end{array}
  \]
  \caption{\oqasm state syntax}
  \label{fig:vqir-state}
\end{figure}

An \oqasm program state is represented according to the grammar in
\Cref{fig:vqir-state}. A state $\varphi$ of $d$ qubits is 
a length-$d$ tuple of qubit values $q$; the state models the tensor
product of those values. This means that the size of $\varphi$ is
$O(d)$ where $d$ is the number of qubits. A $d$-qubit state in a
language like \sqir is represented as a length $2^d$ vector of complex
numbers, $O(2^d)$ in the number of qubits. Our linear state
representation is possible because applying for any well-typed \oqasm
program on any well-formed \oqasm state never causes qubits to be
entangled.

A qubit value $q$ has one of two forms $\overline{q}$, scaled by a
global phase $\alpha(r)$. The two forms depend on the \emph{basis} $\tau$ that the qubit is in---it could be either \texttt{Nor} or \texttt{Phi}. A \texttt{Nor} qubit has form
$\ket{b}$ (where $b \in \{ 0, 1 \}$), which is a
computational basis value. 
A \texttt{Phi} qubit has the form $\qket{r} = \frac{1}{\sqrt{2}}(\ket{0}+\alpha(r)\ket{1})$, which is a value of the (A)QFT basis.
The number $n$ in \texttt{Phi}$\;n$ indicates the precision of the state $\varphi$.
As shown by~Beauregard \cite{qft-adder}, arithmetic on the computational basis can sometimes be more efficiently carried out on the QFT basis, which leads to the use of quantum operations (like QFT) when implementing circuits with classical input/output behavior.

\subsection{OQASM Syntax, Typing, and Semantics}\label{sec:oqasm-syn}

\begin{figure}[t]
\begin{minipage}[t]{0.5\textwidth}
{\small \centering
  \[ \hspace*{-0.8em}
\begin{array}{llcl}
      \text{Position} & p & ::= & (x,n) \qquad   \text{Nat. Num}~n
                                  \qquad   \text{Variable}~x\\
      \text{Instruction} & \instr & ::= & \iskip{p} \mid \inot{p} \mid \irz[\lbrack -1 \rbrack]{n}{p} \mid \iseq{\instr}{\instr}\\
                & & \mid &  \isr[\lbrack -1 \rbrack]{n}{x} \mid \iqft[\lbrack -1 \rbrack]{n}{x} \mid \ictrl{p}{\instr}  \\
                      & & \mid & \ilshift{x} \mid \irshift{x} \mid \irev{x} 
    \end{array}
  \]
}
  \caption{\oqasm syntax. For an operator \texttt{OP}, $\texttt{OP}^{\lbrack -1 \rbrack}$ indicates that the operator has a built-in inverse available.}
  \label{fig:vqir}
\end{minipage}
\hfill
\begin{minipage}[t]{0.38\textwidth}
{\scriptsize
\centering
\begin{tabular}{c@{$\quad=\quad$}c}
  \begin{minipage}{0.3\textwidth}

  \Qcircuit @C=0.5em @R=0.5em {
    \lstick{} & \qw     & \multigate{4}{\texttt{SR m}} & \qw & \qw \\
    \lstick{} & \qw     & \ghost{\texttt{SR m}}           & \qw & \qw \\
    \lstick{} & \vdots & & \vdots & \\
    \lstick{} & & & & \\
    \lstick{} & \qw     & \ghost{\texttt{SR m}}           & \qw  & \qw
    }
  \end{minipage} & 
  \begin{minipage}{0.3\textwidth}
  \small
  \Qcircuit @C=0.5em @R=0.5em {
    \lstick{} & \qw     & \gate{\texttt{RZ (m+1)}} & \qw & \qw \\
    \lstick{} & \qw     & \gate{\texttt{RZ m}}          & \qw & \qw \\
    \lstick{} & & \vdots & & \\
    \lstick{} & & & & \\
    \lstick{} & & & & \\
    \lstick{} & \qw     & \gate{\texttt{RZ 1}}           & \qw  & \qw
    }
 
  \end{minipage} 
\end{tabular}
}
\caption{\texttt{SR} unfolds to a series of \texttt{RZ} instructions}
\label{fig:sr-meaning}
\end{minipage}
\end{figure}

\Cref{fig:vqir} presents \oqasm's syntax. An \oqasm program consists of
a sequence of instructions $\instr$. Each instruction applies an
operator to either a variable $x$, representing a group of qubits, or a \emph{position} $p$, which identifies a particular offset into a variable $x$. 

The instructions in the first row correspond to simple single-qubit
quantum gates---$\iskip{p}$, $\inot{p}$, and $\irz[\lbrack -1 \rbrack]{n}{p}$
 ---and instruction sequencing.
The instructions in the next row apply to whole variables: $\iqft{n}{x}$
applies the AQFT to variable $x$ with $n$-bit precision and
$\iqft[-1]{n}{x}$ applies its inverse.
If $n$ equals the size of $x$, then the AQFT operation is exact.
$\isr[\lbrack -1 \rbrack]{n}{x}$
applies a series of \texttt{RZ} gates (\Cref{fig:sr-meaning}). 
Operation $\ictrl{p}{\instr}$
applies instruction $\instr$ \emph{controlled} on qubit position
$p$. All of the operations in this row---\texttt{SR}, \texttt{QFT}, and \texttt{CU}---will be translated to multiple \sqir
gates. The function \coqe{rz_adder} in \Cref{fig:circuit-exampleb}(b) uses
many of these instructions; e.g., it uses \texttt{QFT} and \texttt{QFT}$^{-1}$ and applies
\texttt{CU} to the $m$th position of variable \texttt{a} to control
instruction \texttt{SR m b}.

In the last row of \Cref{fig:vqir}, instructions $\ilshift{x}$,
$\irshift{x}$, and $\irev{x}$ are \emph{position shifting operations}.
Assuming that $x$ has $d$ qubits and $x_k$ represents the $k$-th qubit
state in $x$, $\texttt{Lshift}\;x$ changes the $k$-th qubit state to
$x_{(k + 1)\mmod d}$, $\texttt{Rshift}\;x$ changes it to
$x_{(k + d - 1)\mmod d}$, and \texttt{Rev} changes it to $x_{d-1-k}$. In
our implementation, shifting is \emph{virtual}, not physical. The \oqasm
translator maintains a logical map of variables/positions to concrete
qubits and ensures that shifting operations are no-ops, introducing no extra gates.

Other quantum operations could be added to \oqasm to
allow reasoning about a larger class of quantum programs while still
guaranteeing a lack of entanglement. 

\begin{figure}[t]
\begin{minipage}[b]{0.57\textwidth}
{\scriptsize
  \begin{mathpar}
    \inferrule[X]{\Omegaty(x)=\texttt{Nor} \\ n < \Omegasz(x)}{\Sigma;\Omega \vdash \inot{(x,n)}\triangleright \Omega}
  
    \inferrule[RZ]{\Omegaty(x)=\texttt{Nor} \\ n < \Omegasz(x)}{\Sigma;\Omega \vdash \irz{q}{(x,n)} \triangleright \Omega}

    \inferrule[SR]{\Omegaty(x)=\tphi{n} \\ m < n}{\Sigma;\Omega \vdash \texttt{SR}\;m\;x\triangleright \Omega}   

    \inferrule[QFT]{\Omegaty(x)=\texttt{Nor}\\n \le \Omegasz(x)}{\Sigma; \Omega \vdash \iqft{n}{x}\triangleright \Omega[x\mapsto \tphi{n}]}    
     
    \inferrule[RQFT]{\Omegaty(x)=\tphi{n}\\n \le \Omegasz(x)}{\Sigma; \Omega \vdash \iqft[-1]{n}{x}\triangleright \Omega[x\mapsto \texttt{Nor}]}             
    
    \inferrule[CU]{\Omegaty(x)=\texttt{Nor} \\ \texttt{fresh}~(x,n)~\instr \\\\ \Sigma; \Omega\vdash \instr\triangleright \Omega \\ \texttt{neutral}(\instr)}{\Sigma; \Omega \vdash \texttt{CU}\;(x,n)\;\instr \triangleright \Omega} 
     
    \inferrule[LSH]{\Omegaty(x)=\texttt{Nor}}{\Sigma; \Omega \vdash \texttt{Lshift}\;x\triangleright \Omega}

     \inferrule[SEQ]{\Sigma; \Omega\vdash \instr_1\triangleright \Omega' \\ \Sigma; \Omega'\vdash \instr_2\triangleright \Omega''}{\Sigma; \Omega \vdash \instr_1\;;\;\instr_2\triangleright \Omega''} 
    
  \end{mathpar}
}
  \caption{Select \oqasm typing rules}
  \label{fig:exp-well-typeda}
\end{minipage}
\hfill
\hfill
\begin{minipage}[b]{0.4\textwidth}
{\footnotesize
\begin{center}\hspace*{-1em}
\begin{tikzpicture}[->,>=stealth',shorten >=1pt,auto,node distance=3.2cm,
                    semithick]
  \tikzstyle{every state}=[fill=black,draw=none,text=white]

  \node[state] (A)              {$\texttt{Nor}$};
  \node[state]         (C) [left of=A] {$\tphi{n}$};

  \path (A) edge [loop above]            node {$\Big\{\begin{array}{l}\texttt{ID},~\texttt{X},~\texttt{RZ}^{\lbrack -1 \rbrack},~\texttt{CU},\\
              \texttt{Rev},\texttt{Lshift},\texttt{Rshift}\end{array}\Big\}$} (A)
            edge   node [above] {\{$\texttt{QFT}\;n$\}} (C);
  \path (C) edge [loop above]            node {$\{\texttt{ID},~\texttt{SR}^{\lbrack -1 \rbrack}\}$} (C)
            edge  [bend right]             node {$\{\texttt{QFT}^{-1}\;n\}$} (A);
\end{tikzpicture}
\end{center}
}
\caption{Type rules' state machine}
\label{fig:state-machine}
\end{minipage}
\end{figure}

\myparagraph{Typing}
\label{sec:vqir-typing}

In \oqasm, typing is concerning a \emph{type environment}
$\Omega$ and a predefined \emph{size
  environment} $\Sigma$, which map \oqasm
variables to their basis and size (number of qubits), respectively.
The typing judgment is written $\Sigma; \Omega\vdash \instr \triangleright \Omega'$ which
states that $\instr$ is well-typed under $\Omega$ and $\Sigma$, and
transforms the variables' bases to be as in $\Omega'$ ($\Sigma$ is unchanged). 
$\Sigma$ is fixed because the number of qubits in execution is always fixed.
It is generated in the high-level language compiler, such as \sourcelang in \cite{oracleoopsla}.
The algorithm generates $\Sigma$ by taking an \sourcelang program and scanning through
all the variable initialization statements.
Select type rules are given in \Cref{fig:exp-well-typeda}; 
the rules not shown (for \texttt{ID}, \texttt{Rshift}, \texttt{Rev}, \texttt{RZ}$^{-1}$, and \texttt{SR}$^{-1}$) are similar.

The type system enforces three invariants. First, it enforces that
instructions are well-formed, meaning that gates are applied to valid
qubit positions (the second premise in \rulelab{X}) and that any control qubit is distinct from the
target(s) (the \texttt{fresh} premise in
\rulelab{CU}).  This latter property enforces the quantum
\emph{no-cloning rule}.
For example, applying the \texttt{CU} in \code{rz\_adder'} (\Cref{fig:circuit-exampleb}) is valid
because position \code{a,m} is distinct from variable \code{b}.

Second, the type system enforces that instructions leave affected
qubits on a proper basis (thereby avoiding entanglement). The
rules implement the state machine shown in
\Cref{fig:state-machine}. For example, $\texttt{QFT}\;n$ transforms a variable from \texttt{Nor} to
$\tphi{n}$ (rule \rulelab{QFT}), while $\texttt{QFT}^{-1}\;n$
transforms it from $\tphi{n}$ back to \texttt{Nor} (rule
\rulelab{RQFT}). Position shifting operations 
are disallowed on variables $x$ in
the \texttt{Phi} basis because the qubits that makeup $x$ are
internally related (see \Cref{appx:well-formed}) and cannot be rearranged. Indeed, applying a
\texttt{Lshift} and then a $\texttt{QFT}^{-1}$ on $x$ in \texttt{Phi}
would entangle $x$'s qubits.

Third, the type system enforces that the effect of position-shifting
operations can be statically tracked. The \texttt{neutral} condition of
\rulelab{CU} requires that any shifting within $\instr$ is restored by the time it
completes. 
For example, $\sseq{\ictrl{p}{(\ilshift{x})}}{\inot{(x,0)}}$ is not well-typed because knowing the final physical position of qubit $(x,0)$ would require statically knowing the value of $p$. 
On the other hand, the program $\sseq{\ictrl{c}{(\sseq{\ilshift{x}}{\sseq{\inot{(x,0)}}{\irshift{x}}})}}{\inot{(x,0)}}$ is well-typed 
because the effect of the \texttt{Lshift} is ``undone'' by an \texttt{Rshift} inside the body of the \texttt{CU}.

\myparagraph{Semantics}\label{sec:pqasm-dsem}

\begin{figure}[t]
{\scriptsize
\[
\begin{array}{lll}
\llbracket \iskip{p} \rrbracket\varphi &= \varphi\\[0.2em]

\llbracket \inot{(x, i)} \rrbracket\varphi &= \app{\uparrow\xsem(\downarrow\varphi(x,i))}{\varphi}{(x,i)}
& \texttt{where  }\xsem(\ket{0})=\ket{1} \qquad\, \xsem(\ket{1})=\ket{0}
\\[0.5em]

\llbracket \ictrl{(x,i)}{\instr} \rrbracket\varphi &=  \csem(\downarrow\varphi(x,i),\instr,\varphi)
&
\texttt{where  }
\csem({\ket{0}},{\instr},\varphi)=\varphi\quad\;\,
\csem({\ket{1}},{\instr},\varphi)=\llbracket \instr \rrbracket\varphi
\\[0.4em]

\llbracket \irz{m}{(x,i)} \rrbracket\varphi &= \app{\uparrow {\rsem}({m},\downarrow\varphi(x,i))}{\varphi}{(x,i)}
&\texttt{where  }{\rsem}(m,\ket{0})=\ket{0} \; \quad{\rsem}(m,\ket{1})=\alpha(\frac{1}{2^m})\ket{1}
\\[0.5em]

\llbracket \irz[-1]{m}{(x,i)} \rrbracket\varphi &= \app{\uparrow {\rrsem}({m},\downarrow\varphi(x,i))}{\varphi}{(x,i)}
 &\texttt{where  }{\rrsem}(m,\ket{0})=\ket{0}
\quad{\rrsem}(m,\ket{1})=\alpha(-\frac{1}{2^m})\ket{1}
\\[0.5em]

\llbracket \isr{m}{x} \rrbracket\varphi &
                                            \multicolumn{2}{l}{= \app{\uparrow \qket{r_i+\frac{1}{2^{m-i+1}}}}{\varphi}{\forall i \le m.\;(x,i)}
\qquad \texttt{when  }
\downarrow\varphi(x,i) = \qket{r_i}}\\[0.5em]

\llbracket \isr[-1]{m}{x} \rrbracket\varphi&\multicolumn{2}{l}{= \app{\uparrow \qket{r_i-\frac{1}{2^{m-i+1}}}}{\varphi}{\forall i \le m.\;(x,i)}
\qquad \texttt{when  }
\downarrow\varphi(x,i) = \qket{r_i}}\\[0.5em]

\llbracket \iqft{n}{x} \rrbracket\varphi &= \app{\uparrow\qsem(\Sigma(x),\downarrow\varphi(x),n)}{\varphi}{x}
& \texttt{where  }\qsem(i,\ket{y},n)=\bigotimes_{k=0}^{i-1}(\qket{\frac{y}{2^{n-k}}})
\\[0.5em]

\llbracket \iqft[-1]{n}{x} \rrbracket\varphi &=  \app{\uparrow\qsem^{-1}(\Sigma(x),\downarrow\varphi(x),n)}{\varphi}{x}
\\[0.5em]

\llbracket \ilshift{x} \rrbracket\varphi &= \app{{\psem}_{l}(\varphi(x))}{\varphi}{x}
&
\texttt{where  }{\psem}_{l}(q_0\otimes q_1\otimes \cdots \otimes q_{n-1})=q_{n-1}\otimes q_0\otimes q_1 \otimes \cdots
\\[0.5em]

\llbracket \irshift{x} \rrbracket\varphi &= \app{{\psem}_{r}(\varphi(x))}{\varphi}{x}
&
\texttt{where  }{\psem}_{r}(q_0\otimes q_1\otimes \cdots \otimes q_{n-1})=q_1\otimes \cdots \otimes q_{n-1} \otimes q_0
\\[0.5em]

\llbracket \irev{x} \rrbracket\varphi &= \app{{\psem}_{a}(\varphi(x))}{\varphi}{x}
&
\texttt{where  }{\psem}_{a}(q_0\otimes \cdots \otimes q_{n-1})=q_{n-1}\otimes \cdots \otimes q_0
\\[0.5em]

\llbracket \iota_1; \iota_2 \rrbracket\varphi &= \llbracket \iota_2 \rrbracket (\llbracket \iota_1 \rrbracket\varphi)
\end{array}
\]
}
{\scriptsize
$
\begin{array}{l}
\\[0.2em]
\downarrow \alpha(b)\overline{q}=\overline{q}
\qquad
\downarrow (q_1\otimes \cdots \otimes q_n) = \downarrow q_1\otimes \cdots \otimes \downarrow q_n
\\[0.2em]
\app{\uparrow \overline{q}}{\varphi}{(x,i)}=\app{\alpha(b)\overline{q}}{\varphi}{(x,i)}
\qquad \texttt{where  }\varphi(x,i)=\alpha(b)\overline{q_i}
\\[0.2em]
\app{\uparrow \alpha(b_1)\overline{q}}{\varphi}{(x,i)}=\app{\alpha(b_1+b_2)\overline{q}}{\varphi}{(x,i)}
\qquad \texttt{where  }\varphi(x,i)=\alpha(b_2)\overline{q_i}
\\[0.2em]
\app{q_x}{\varphi}{x}=\app{q_{(x,i)}}{\varphi}{\forall i < \Sigma(x).\;(x,i)}
\\[0.2em]
\app{\uparrow q_x}{\varphi}{x}=\app{\uparrow q_{(x,i)}}{\varphi}{\forall i < \Sigma(x).\;(x,i)}
\end{array}
$
}
\caption{\oqasm semantics}
  \label{fig:deno-sema}
\end{figure}

The semantics of an \oqasm program is a partial function
$\llbracket\rrbracket$ from
an instruction $\instr$ and input state $\varphi$ to an output state
$\varphi'$, written 
$\llbracket \instr \rrbracket\varphi=\varphi'$, shown in \Cref{fig:deno-sema}.

Recall that a state $\varphi$ is a tuple of $d$ qubit values,
modeling the tensor product $q_1\otimes \cdots \otimes q_d$. 
The rules implicitly map each variable $x$ to a
range of qubits in the state, e.g., 
$\varphi(x)$ corresponds to some sub-state $q_k\otimes \cdots \otimes q_{k+n-1}$
where $\Omegasz(x)=n$.
Many of the rules in \Cref{fig:deno-sema} update a \emph{portion} of a
state. $\app{q_{(x,i)}}{\varphi}{(x,i)}$ updates the $i$-th
qubit of variable $x$ to be the (single-qubit) state $q_{(x,i)}$, and
$\app{q_{x}}{\varphi}{x}$ to update variable $x$ according to
the qubit \emph{tuple} $q_x$.
$\app{\uparrow q_{(x,i)}}{\varphi}{(x,i)}$ and $\app{\uparrow q_{x}}{\varphi}{x}$ are similar, except that they also accumulate the previous global phase of $\varphi(x,i)$ (or $\varphi(x)$).
$\downarrow$ is to convert a qubit $\alpha(b)\overline{q}$ to an unphased qubit $\overline{q}$.

Function $\xsem$ updates the state of a single
qubit according to the rules for the standard quantum gate $X$.  
\texttt{cu} is a conditional operation
depending on the \texttt{Nor}-basis qubit $(x,i)$. 
\texttt{RZ} (or $\texttt{RZ}^{-1}$) is an z-axis phase rotation operation.
Since it applies to \texttt{Nor}-basis, it applies a global phase.
By \Cref{thm:sem-same}, when it is compiled to \sqir,
the global phase might be turned into a local one.
For example, to prepare the state $\sum_{j=0}^{2^n\sminus 1}(-i)^x\ket{x}$ \cite{ChildsNAND}, 
a series of Hadamard gates are applied, followed by several controlled-\texttt{RZ} gates on $x$,
where the controlled-\texttt{RZ} gates are definable by \oqasm.
\texttt{SR} (or
$\texttt{SR}^{-1}$) applies an $m+1$ series of \texttt{RZ} (or
$\texttt{RZ}^{-1}$) rotations where the $i$-th rotation
applies a phase of $\alpha({\frac{1}{2^{m-i+1}}})$
(or $\alpha({-\frac{1}{2^{m-i+1}}})$).
$\qsem$ applies an approximate quantum Fourier transform; $\ket{y}$ is an abbreviation of
$\ket{b_1}\otimes \cdots \otimes \ket{b_i}$ (assuming $\Omegasz(y)=i$) and $n$ is the degree of approximation.
If $n = i$, then the operation is the standard QFT\@.
Otherwise, each qubit in the state is mapped to $\qket{\frac{y}{2^{n-k}}}$, which is equal to $\frac{1}{\sqrt{2}}(\ket{0} + \alpha(\frac{y}{2^{n-k}})\ket{1})$ when $k < n$ and $\frac{1}{\sqrt{2}}(\ket{0} + \ket{1}) = \ket{+}$ when $n \leq k$ (since $\alpha(n) = 1$ for any natural number $n$).
$\qsem^{-1}$ is the inverse function of $\qsem$. 
Note that the input state to $\qsem^{-1}$ is guaranteed to have the form $\bigotimes_{k=0}^{i-1}(\qket{\frac{y}{2^{n-k}}})$ because it has type $\tphi{n}$.
$\psem_l$, $\psem_r$, and
$\psem_a$ are the semantics for \itext{Lshift}, 
\itext{Rshift}, and \itext{Rev}, respectively.   

\subsection{OQASM Metatheory}\label{sec:metatheory}

\myparagraph{Soundness}
The following statement is proved: well-typed \oqasm programs are well-defined; i.e., the type system is sound concerning the semantics. 
Below is the well-formedness of an \oqasm state.

\begin{definition}[Well-formed \oqasm state]\label{appx:well-formed}\rm 
  A state $\varphi$ is \emph{well-formed}, written
  $\Sigma;\Omega \vdash \varphi$, iff:
\begin{itemize}
\item For every $x \in \Omega$ such that $\Omegaty(x) = \texttt{Nor}$,
  for every $k <\Omegasz(x)$, $\varphi(x,k)$ has the form
  $\alpha(r)\ket{b}$.

\item For every $x \in \Omega$ such that $\Omegaty(x) = \tphi{n}$ and $n \le \Omegasz(x)$,
  there exists a value $\upsilon$ such that for
  every $k < \Omegasz(x)$, $\varphi(x,k)$ has the form
  $\alpha(r)\qket{\frac{\upsilon}{ 2^{n- k}}}$.\footnote{Note that $\Phi(x) = \Phi(x + n)$, where the integer $n$ refers to phase $2 \pi n$; so multiple choices of $\upsilon$ are possible.}
\end{itemize}
\end{definition}

\noindent
Type soundness is stated as follows; the proof is by induction on $\instr$ and is mechanized in Rocq.

\begin{theorem}\label{thm:type-sound-oqasm}\rm[\oqasm type soundness]
If $\Sigma; \Omega \vdash \instr \triangleright \Omega'$ and $\Sigma;\Omega \vdash \varphi$ then there exists $\varphi'$ such that $\llbracket \instr \rrbracket\varphi=\varphi'$ and $\Sigma;\Omega' \vdash \varphi'$.
\end{theorem}

\myparagraph{Algebra}
Mathematically, the set of well-formed $d$-qubit \oqasm states for a given $\Omega$ can be interpreted as a subset $\hsp{S}^d$ of a $2^d$-dimensional Hilbert space $\hsp{H}^d$ \footnote{A \emph{Hilbert space} is a vector space with an inner product that is complete with respect to the norm defined by the inner product. $\hsp{S}^d$ is a sub\emph{set}, not a sub\emph{space} of $\hsp{H}^d$ because $\hsp{S}^d$ is not closed under addition: Adding two well-formed states can produce a state that is not well-formed.}. The semantics function $\llbracket \rrbracket$ can be interpreted as a $2^d \times 2^d$ unitary matrix, as is standard when representing the semantics of programs without measurement~\cite{PQPC}.
Because \oqasm's semantics can be viewed as a unitary matrix, correctness properties extend by linearity from $\hsp{S}^d$ to $\hsp{H}^d$---an oracle that performs addition for classical \texttt{Nor} inputs will also perform addition over a superposition of \texttt{Nor} inputs. The following statement is proved: $\hsp{S}^d$ is closed under well-typed \oqasm programs.

Given a qubit size map $\Sigma$ and type environment $\Omega$, the set of \oqasm programs that are well-typed concerning $\Sigma$ and $\Omega$ (i.e., $\Sigma;\Omega \vdash \instr \triangleright \Omega'$) form an algebraic structure $(\{\instr\},\Sigma, \Omega,\hsp{S}^d)$, where $\{\instr\}$ defines the set of valid program syntax, such that there exists $\Omega'$, $\Sigma;\Omega \vdash \instr \triangleright \Omega'$ for all $\instr$ in $\{\instr\}$; $\hsp{S}^d$ is the set of $d$-qubit states on which programs $\instr\in \{\instr\}$ are run, and are well-formed ($\Sigma;\Omega \vdash \varphi$) according to \Cref{appx:well-formed}.
From the \oqasm semantics and the type soundness theorem, for all $\instr \in \{\instr\}$ and $\varphi \in \hsp{S}^d$, such that $\Sigma;\Omega \vdash \instr \triangleright \Omega'$ and $\Sigma;\Omega \vdash \varphi$, and $\llbracket \instr \rrbracket\varphi=\varphi'$, $\Sigma;\Omega' \vdash \varphi'$, and $\varphi' \in \hsp{S}^d$. Thus, $(\{\instr\},\Sigma, \Omega,\hsp{S}^d)$, where $\{\instr\}$ defines a groupoid.

The groupoid can be certainly extended to another algebraic structure $(\{\instr'\},\Sigma,\hsp{H}^d)$, where $\hsp{H}^d$ is a general $2^d$ dimensional Hilbert space $\hsp{H}^d$ and $\{\instr'\}$ is a universal set of quantum gate operations.
Clearly, the following is true: $\hsp{S}^d \subseteq \hsp{H}^d$ and $\{\instr\} \subseteq \{\instr'\}$, because sets $\hsp{H}^d$ and $\{\instr'\}$ can be acquired by removing the well-formed ($\Sigma;\Omega \vdash \varphi$) and well-typed ($\Sigma;\Omega \vdash \instr \triangleright \Omega'$) definitions for $\hsp{S}^d$ and $\{\instr\}$, respectively.
$(\{\instr'\},\Sigma,\hsp{H}^d)$ is a groupoid because every \oqasm operation is valid in a traditional quantum language like \sqir. The following two theorems are to connect \oqasm operations with operations in the general Hilbert space: 

 \begin{theorem}\label{thm:subgroupoid}\rm
   $(\{\instr\},\Sigma, \Omega,\hsp{S}^d) \subseteq (\{\instr\},\Sigma,\hsp{H}^d)$ is a subgroupoid.
 \end{theorem}

\begin{theorem}\label{thm:sem-same}\rm
Let $\ket{y}$ be an abbreviation of $\bigotimes_{m=0}^{d-1} \alpha(r_m) \ket{b_m}$ for $b_m \in \{0,1\}$.
If for every $i\in [0,2^d)$, $\llbracket \instr \rrbracket\ket{y_i}=\ket{y'_i}$, then $\llbracket \instr \rrbracket (\sum_{i=0}^{2^d-1} \ket{y_i})=\sum_{i=0}^{2^d-1} \ket{y'_i}$.
\end{theorem}

The following theorems are proved as corollaries of the compilation correctness theorem from \oqasm to \sqir (\cite{oracleoopsla}). 
\Cref{thm:subgroupoid} suggests that the space $\hsp{S}^d$ is closed under the application of any well-typed \oqasm operation.
\Cref{thm:sem-same} says that \oqasm oracles can be safely applied to superpositions over classical states.\footnote{Note that a superposition over classical states can describe \emph{any} quantum state, including entangled states.}

\begin{figure}[t]
  {\scriptsize
    \begin{mathpar}
      \inferrule[ ]{}{\inot{(x,n)}\xrightarrow{\text{inv}} \inot{(x,n)}}
    
      \inferrule[  ]{}{\texttt{SR}\;m\;x\xrightarrow{\text{inv}} \texttt{SR}^{-1}\;m\;x}
  
      \inferrule[ ]{}{\iqft{n}{x} \xrightarrow{\text{inv}}  \iqft[-1]{n}{x}}   
  
      \inferrule[ ]{}{\texttt{Lshift}\;x\xrightarrow{\text{inv}} \texttt{Rshift}\;x} 
       
      \inferrule[ ]{\instr \xrightarrow{\text{inv}} \instr'}{\texttt{CU}\;(x,n)\;\instr \xrightarrow{\text{inv}} \texttt{CU}\;(x,n)\;\instr'} 
  
      \inferrule[ ]{\instr_1 \xrightarrow{\text{inv}} \instr'_1 \\ \instr_2 \xrightarrow{\text{inv}} \instr'_2}{\instr_1\;;\;\instr_2\xrightarrow{\text{inv}} \instr'_2\;;\;\instr'_1} 
      
    \end{mathpar}
  }
  \caption{Select \oqasm inversion rules}
  \label{fig:exp-reversed-fun}
\end{figure}

\begin{figure}[t]
{\tiny
\hspace*{-2em}
\centering
\begin{tabular}{c@{$\quad=\quad$}c@{\qquad}c@{$\quad=\quad$}c}
  \begin{minipage}{0.25\textwidth}
  \footnotesize
  \Qcircuit @C=0.25em @R=0.35em {
    & \qw & \multigate{3}{(x+a)_n} & \qw \\
    & \vdots & & \\
    & & & \\
    & \qw & \ghost{(x+a)_n} & \qw \\
    }
  \end{minipage}
&
\begin{minipage}{.45\textwidth}
  \footnotesize
  \Qcircuit @C=0.35em @R=0.55em {
     & \qw & \gate{\texttt{SR}\;0} & \multigate{3}{\texttt{SR}\;1} & \qw & \qw & \qw & \multigate{5}{\texttt{SR}\;(n-1)} & \qw  \\
      & & & & & \dots & & &  \\
      & \qw & \qw  &  \ghost{\texttt{SR}\; 1} & \qw & \qw & \qw & \ghost{\texttt{SR}\;(n-1)} & \qw \\
      & & & & & & & &  \\
     & & & & & & & &  \\
   & \qw & \qw & \qw & \qw & \qw & \qw & \ghost{\texttt{SR}\;(n-1)}  & \qw 
    }
\end{minipage}
&  
\begin{minipage}{0.25\textwidth}
  \footnotesize
  \Qcircuit @C=0.25em @R=0.35em {
    & \qw & \multigate{3}{(x-a)_n} & \qw \\
    & \vdots & & \\
    & & & \\
    & \qw & \ghost{(x+a)_n} & \qw \\
    }
  \end{minipage}
&
\begin{minipage}{.45\textwidth}
  \footnotesize
  \Qcircuit @C=0.35em @R=0.55em {
    & \qw & \multigate{5}{\texttt{SR}^{-1} (n-1)} & \qw & \qw & \qw & \multigate{3}{\texttt{SR}^{-1} 1} & \gate{\texttt{SR}^{-1} 0} & \qw \\
    &     &                                  &     & \dots &   &                              &                      &   \\
    & \qw & \ghost{\texttt{SR}^{-1} (n-1)}        & \qw & \qw   & \qw & \ghost{\texttt{SR}^{-1} 1} & \qw & \qw  \\
      & & & & & & & &  \\
     & & & & & & & &  \\
    & \qw & \ghost{\texttt{SR}^{-1} (n-1)} & \qw & \qw & \qw & \qw & \qw & \qw 
    }
\end{minipage}
\end{tabular}
}
\caption{Addition/subtraction circuits are inverses}
\label{fig:circuit-add-sub}
\end{figure}

\oqasm programs are easily invertible, as shown by the rules in \Cref{fig:exp-reversed-fun}.
This inversion operation is useful for constructing quantum oracles; for example, the core logic in the QFT-based subtraction circuit is just the inverse of the core logic in the addition circuit (\Cref{fig:exp-reversed-fun}).
This allows us to reuse the proof of addition in the proof of subtraction.
The inversion function satisfies the following properties:

 \begin{theorem}\label{thm:reversibility}\rm[Type reversibility]
    For any well-typed program $\instr$, such that $\Sigma; \Omega \vdash \instr \triangleright \Omega'$, its inverse $\instr'$, where $\instr \xrightarrow{\text{inv}} \instr'$, is also well-typed and $\Sigma;\Omega' \vdash \instr' \triangleright \Omega$. Moreover, $\llbracket \instr ; \instr' \rrbracket \varphi=\varphi$.
 \end{theorem}

\end{document}